%% file: classifieds.tex
\documentclass[acmsmall, authorversion, screen]{acmart}

\setcopyright{rightsretained}
\acmPrice{}
\acmDOI{10.1145/3290333}
\acmYear{2019}
\copyrightyear{2019}
\acmJournal{PACMPL}
\acmVolume{3}
\acmNumber{POPL}
\acmMonth{1}

\usepackage{booktabs}   %% For formal tables:
\usepackage{subcaption} %% For complex figures with subfigures/subcaptions
\usepackage{epigraph}
\usepackage{prooftree}

\usepackage{amsmath}
\usepackage{amsthm}
\usepackage{amssymb}
\usepackage{mathtools}
\usepackage{upgreek}
\usepackage{tikz-cd}
\tikzset{%
    symbol/.style={%
        draw=none,
        every to/.append style={%
            edge node={node [sloped, allow upside down, auto=false]{$#1$}}}
    }
}

\tikzset{%
    symbol/.style={%
        draw=none,
        every to/.append style={%
            edge node={node [sloped, allow upside down, auto=false]{$#1$}}}
    }
}

\newtheorem{thm}{Theorem}
\newtheorem{lem}{Lemma}
\newtheorem{cor}{Corollary}
\newtheorem{prop}{Proposition}
\newtheorem{defn}{Definition}

\newcommand{\setcomp}[2]{\left\{ \, #1 \; \middle| \; #2 \, \right\}} 
\newcommand{\sem}[2]{\left\llbracket #1 \right\rrbracket^{#2}}
\newcommand{\myeq}{\stackrel{\mathclap{\mbox{\tiny def}}}{=}}

\newcommand{\ctxt}[2]{#1\mathbin{|}#2}
\newcommand{\ibox}[1]{\mathsf{box\;}#1}
\newcommand{\letbox}[3]{\mathsf{let\;box\;} #1 = #2 \mathsf{\;in\;} #3}

\newcommand{\bars}[1]{\left\lvert #1 \right\rvert}

\newcommand{\rel}[1]{\mathrel{R}_{#1}}

\begin{document}

\title{Modalities, Cohesion, and Information Flow}

\author{G. A. Kavvos}
\orcid{0000-0001-7953-7975}
\affiliation{
  \position{Postdoctoral Research Associate}
  \department{Department of Mathematics and Computer Science}
  \institution{Wesleyan University}
  \streetaddress{265 Church Street}
  \city{Middletown}
  \state{Connecticut}
  \postcode{06459}
  \country{United States of America}
}
\email{gkavvos@wesleyan.edu}

\begin{abstract}
  It is informally understood that the purpose of modal type
  constructors in programming calculi is to control the flow of
  information between types. In order to lend rigorous support to
  this idea, we study the category of classified sets, a variant
  of a denotational semantics for information flow proposed by
  Abadi et al. We use classified sets to prove multiple
  noninterference theorems for modalities of a monadic and
  comonadic flavour. The common machinery behind our theorems
  stems from the the fact that classified sets are a (weak) model
  of Lawvere's theory of axiomatic cohesion. In the process, we
  show how cohesion can be used for reasoning about multi-modal
  settings. This leads to the conclusion that cohesion is a
  particularly useful setting for the study of both information
  flow, but also modalities in type theory and programming
  languages at large.
\end{abstract}

%% 2012 ACM Computing Classification System (CSS) concepts
%% Generate at 'http://dl.acm.org/ccs/ccs.cfm'.
\begin{CCSXML}
<ccs2012>
<concept>
<concept_id>10003752.10003790.10003793</concept_id>
<concept_desc>Theory of computation~Modal and temporal logics</concept_desc>
<concept_significance>500</concept_significance>
</concept>
<concept>
<concept_id>10003752.10003790.10011740</concept_id>
<concept_desc>Theory of computation~Type theory</concept_desc>
<concept_significance>500</concept_significance>
</concept>
<concept>
<concept_id>10003752.10010124.10010131.10010133</concept_id>
<concept_desc>Theory of computation~Denotational semantics</concept_desc>
<concept_significance>500</concept_significance>
</concept>
<concept>
<concept_id>10003752.10010124.10010131.10010137</concept_id>
<concept_desc>Theory of computation~Categorical semantics</concept_desc>
<concept_significance>500</concept_significance>
</concept>
<concept>
<concept_id>10003752.10010124.10010125.10010130</concept_id>
<concept_desc>Theory of computation~Type structures</concept_desc>
<concept_significance>300</concept_significance>
</concept>
<concept>
<concept_id>10011007.10011006.10011008.10011009.10011012</concept_id>
<concept_desc>Software and its engineering~Functional languages</concept_desc>
<concept_significance>500</concept_significance>
</concept>
<concept>
<concept_id>10011007.10011006.10011008</concept_id>
<concept_desc>Software and its engineering~General programming languages</concept_desc>
<concept_significance>300</concept_significance>
</concept>
</ccs2012>
\end{CCSXML}

\ccsdesc[500]{Theory of computation~Modal and temporal logics}
\ccsdesc[500]{Theory of computation~Type theory}
\ccsdesc[500]{Theory of computation~Denotational semantics}
\ccsdesc[500]{Theory of computation~Categorical semantics}
\ccsdesc[300]{Theory of computation~Type structures}
\ccsdesc[500]{Software and its engineering~Functional languages}
\ccsdesc[300]{Software and its engineering~General programming languages}
%% End of generated code

%% Keywords
%% comma separated list
\keywords{information flow, information flow control, type
systems, modal type systems, cohesion, modal type theory,
modalities, noninterference, category theory}

\maketitle

\section{Introduction}

Taming the flow of information within a computer system has been a
problem of significant interest since the early days of Computer
Science; see, for example, the models of Bell and LaPadula
\cite{Bell1996, Rushby1986}, and the influential work of
\citet{Denning1976}, which was the first to introduce the use of
\emph{lattices} for modelling secure information flow. The
objective of these models is usually to express the property that
data cannot flow in or out of certain regions of a computer
system, thus achieving a certain form of confidentiality or
integrity.

A modern way of achieving the above objective is to make it
\emph{language-based}. That is: to enrich a programming language
with features that specify or control the flow of data, so that
the programs we write are correct---or rather, secure---by design.
One way of doing so is through the use of \emph{type systems}
which annotate variables or program expressions with security
levels. It then suffices to prove that the type system ensures
some form of \emph{noninterference property}, which invariably
states that, in a well-typed program, data cannot flow contrary to
our wishes, e.g. from a type labelled as being of high security to
one of low security. Several type systems of this form have been
proposed, but all seem to be in some sense `equivalent;' see e.g.
the recent work of \citet{Rajani2018}.

Many of theses type systems feature some form of \emph{modality},
which is broadly construed as a unary type constructor of some
sort: see e.g. \cite{Abadi1999, Miyamoto2004, Shikuma2008}. The
rules that govern the behaviour of a modality almost always
silently endow it with certain information flow properties.  This
is intuitively well-known both by modal type theorists, as well as
practitioners who use modal type systems in programming calculi.

Nevertheless, the implicit properties of these systems have not
been subjected to a detailed treatment. Given the recent
resurgence of interest in modal type theory---as exemplified by
cohesive homotopy type theory \cite{Shulman2018}, and guarded type
theory \cite{Clouston2016} on the theoretical side, but also
calculi for functional reactive programming
\cite{Krishnaswami2013}, effects \cite{Curien2016} or coeffects
\cite{Petricek2014} on the more application-driven end---we
believe that there is a need for a more universal approach to
modal type theory. There have been major recent advances on the
syntactical side, particularly through the fibrational framework
of \citet{Licata2017}. However, we are still lacking a
`Swiss-army-knife model' that can help us understand and prove
properties about the information flow of these modalities. This is
the subject of the present paper.

For this purpose we introduce a refinement of the \emph{dependency
category} of \citet{Abadi1999}, which we call \emph{the category
of classified sets}. These are sets equipped with indexed logical
relations that encode \emph{indistinguishability}, and thus a form
of \emph{data hiding}. \citet{Abadi1999} used a variant of this
model to prove noninterference theorems for various information
flow calculi that they translated to their \emph{dependency core
calculus}, 

We will present a different, significantly more `high tech'
approach. First, we notice that the relations of
indistinguishability that classified sets carry must be respected:
that is, if $x \rel{} y$ then $f(x) \rel{} f(y)$ for any morphism
$f : X \rightarrow Y$ of our model. Thus, if $x \rel{} y$, the
`points' $x$ and $y$ of $X$ can be considered to be `arbitrarily
close' to each other in a `space.' So close, in fact, as to not be
distinguishable. This analogy is rife with topological intuition.
We will see that these relations describe a sort of
\emph{cohesion} between points. It so happens that an axiomatic
approach to this idea has been developed by \citet{Lawvere2007},
and that classified sets form a weak model of this theory. 

Once this fact is established, we can show that many of the
information flow properties that we wish to establish follow
directly from this abstract framework of cohesion. From that
point, we want to (a) identify appropriate structure in the
category of classified sets for modelling a host of modal type
theories that are used in information flow control, and (b) use
this structure to prove noninterference theorems for these type
theories.

This paper proceeds as follows. In \S\ref{sec:defns} we introduce
the category of classified sets, and show that it is a finitely
complete and cocomplete, bicartesian closed category, and hence a
model of the simply typed $\lambda$-calculus.  Then, in
\S\ref{sec:coh} we will present the rudiments of \emph{axiomatic
cohesion}, and show that classified sets are pre-cohesive relative
to ordinary sets. This will introduce some modal operations, and
lead us to prove some basic noninterference theorems in
\S\ref{sec:noninterference1}. 

Then, in \S\ref{sec:coh2}, we present a levelled view of cohesion.
Classified sets are defined parametrically in a set of security
levels $\mathcal{L}$. If we add a new set $\pi$ of security
levels, then the new category of sets---which is classified over
$\mathcal{L} \cup \pi$---is pre-cohesive over the category of sets
classified over $\mathcal{L}$. This generates some more modal
operators, which are now `level-sensitive.' These are used in
\S\ref{sec:noninterference2} to prove another set of
noninterference theorems. We make some concluding remarks in
\S\ref{sec:conc}.

Category theory is used extensively throughout the paper. The main
tool is that of \emph{adjunctions}, and their close relationship
to (co)monads, both of which are beautifully covered in \cite[\S
9-10]{Awodey2010}. We only use one advanced concept, namely that
of (co)reflective subcategories, which correspond to
\emph{idempotent} monads and comonads; this is covered in \cite[\S
IV.3]{MacLane1978}, \cite[Vol. II, \S 4.2.4]{Borceux1994}, or the
nLab
wiki.\footnote{\url{https://ncatlab.org/nlab/show/reflective+subcategory}}

\section{Classified Sets}
  \label{sec:defns}

\setlength{\epigraphwidth}{0.7\textwidth}

\epigraph{Reynolds integrated these two strands of thought and
formulated a general principle of relational parametricity that is
applicable to a wide range of contexts for capturing the notion of
``information hiding'' or ``abstraction.'' Unfortunately, we
believe that the magnitude of this achievement has not been
sufficiently recognized.}{C. Hermida, U. S. Reddy, and E. P.
Robinson \cite{Hermida2014}}

Let $\mathcal{L}$ be a set of labels, which we call \emph{security
levels}. We assume precisely nothing about $\mathcal{L}$, so our
theory is curiously independent of its structure (finite,
infinite, partial order, lattice, etc.).

\begin{defn}
  A \emph{classified set $S$} over $\mathcal{L}$ (or: a set $S$
  classified over $\mathcal{L}$) consists of
  \begin{enumerate}
      \item
        an ordinary carrier set $\bars{S}$, and
      \item
        a family of reflexive relations $(R_\ell)_{\ell \in
        \mathcal{L}}$ on $\bars{S}$, one for each level $\ell \in
        \mathcal{L}$.
    \end{enumerate}
\end{defn}

\noindent We will---more often than not---write $x \in S$ to mean
$x \in \bars{S}$. The underlying intuition pertaining to a
classified set is that each level $\ell \in \mathcal{L}$ is to be
understood as a \emph{security clearance}, and the relation
$\rel{\ell}$ models \emph{indistinguishability} for users at that
clearance. That is: if $x \rel{\ell} y$, then a user with
clearance $\ell$ must \emph{not} be able to distinguish between
$x$ and $y$. Reflexivity models the simple fact that $x$ should be
indistinguishable to itself. In logical relations, reflexivity is
a theorem; but since not everything is defined inductively here,
it must become an explicit requirement.\footnote{\citet{Abadi1999}
did not require their relations to be reflexive, which is a key
property in showing pre-cohesion in \S\ref{sec:coh}.}

The kind of \emph{functions} admissible in our mathematical
universe shall be precisely those that map indistinguishable
inputs to indistinguishable outputs.

\begin{defn} 
  Let $S$ and $S'$ be sets classified over $\mathcal{L}$. A
  \emph{morphism of classified sets} $f : S \rightarrow S'$ is a
  function $f : \bars{S} \rightarrow \bars{S'}$ such that $x
  \rel{\ell} y$ implies $f(x) \rel{\ell} f(y)$ for all $\ell \in
  \mathcal{L}$.
\end{defn}

\noindent Classified sets over $\mathcal{L}$ and their morphisms
constitute a category, which we denote as
$\textbf{CSet}_\mathcal{L}$.

\subsection{Limits and Colimits}

We move on to the examination of what kind of limits and colimits
exist in classified sets, which tells us which kinds of data type
we are able to construct.

We let $\mathbf{1}$ be the classified set with a singleton carrier
set $\{\ast\}$, and $\ast \rel{\ell} \ast$ for all $\ell \in
\mathcal{L}$.

\begin{prop}
  $\mathbf{1}$ is a terminal object in
  $\textbf{CSet}_\mathcal{L}$.
\end{prop} That is: there is a unique function from a classified
set $X$ to $\mathbf{1}$; it maps everything to $\ast$, hence
collapsing all related pairs to one element, which cannot be
distinguished from itself.

Similarly, we let the classified set $\mathbf{0}$ be the set whose
carrier is the empty set, and all of whose relations are empty.
Since there are no relations to preserve, it is evident that the
unique empty function from $\mathbf{0}$ to any $\bars{A}$ is a
morphism, so

\begin{prop}
  $\mathbf{0}$ is an initial object in
  $\textbf{CSet}_\mathcal{L}$.
\end{prop}

For classified sets $A$ and $B$, we let $\bars{A \times B} \myeq
\bars{A} \times \bars{B}$. Given any $\ell \in \mathcal{L}$, we
define \[
  (a_1, b_1) \rel{\ell} (a_2, b_2) 
    \Longleftrightarrow
  a_1 \rel{\ell} a_2 \wedge b_1 \rel{\ell} b_2
\] Hence, two pairs are indistinguishable exactly when they are so
componentwise. If, for example, the second components are
distinguishable, then so are the pairs; but this does not entail
that we can distinguish the first components! In this way, we can
`classify' pairs without resorting to more complicated sets of
labels, e.g. $\mathcal{L} \times \mathcal{L}$.

The standard set-theoretic projections preserve $\rel{\ell}$, as
does the standard set-theoretic product morphism $\langle f, g
\rangle : \bars{C} \rightarrow \bars{A} \times \bars{B}$ for any
$f : C \rightarrow A$ and $g : C \rightarrow B$. Hence,
\begin{prop}
  The classified set $A \times B$ is the categorical product of
  $A$ and $B$ in $\textbf{CSet}_\mathcal{L}$.
\end{prop}

\noindent A similar story applies to coproducts: given $A$ and
$B$, we define \[
  \bars{A + B} \myeq \bars{A} + \bars{B}
    = \setcomp{(0, a)}{a \in \bars{A}} \cup
      \setcomp{(1, b)}{b \in \bars{B}}
\] and, naturally, we define $\rel{\ell}$ to be \[
  (0, a) \rel{\ell} (0, a')
    \Leftrightarrow a \rel{\ell} a',
    \quad
  (1, b) \rel{\ell} (1, b')
    \Leftrightarrow b \rel{\ell} b'
\] and $\lnot\left((i, x) \rel{\ell} (j, y)\right)$ for $i \neq
j$.  The injections $\bars{A} \rightarrow \bars{A} + \bars{B}$ and
$\bars{B} \rightarrow \bars{A} + \bars{B}$ clearly preserve
$\rel{\ell}$, as does the set-theoretic coproduct morphism, so
\begin{prop}
  The classified set $A + B$ is the categorical coproduct of
  $A$ and $B$ in $\textbf{CSet}_\mathcal{L}$.
\end{prop}

In a similar fashion, if we are given two parallel arrows
$\begin{tikzcd} A \arrow[r, "f", shift left] \arrow[r, "g", swap,
shift right] & B \end{tikzcd}$, we can see that the set $E \myeq
\setcomp{a \in A}{f(a) = g(a)}$ equipped with
$\rel{\ell}\restriction_E$, the relations $\rel{\ell}$ restricted
to $E$, is a classified set, that the inclusion $E \hookrightarrow
A$ is trivially a morphism, and that

\begin{prop}
  $E$ is the equaliser of $f: A \rightarrow B$ and $g : A
  \rightarrow B$.
\end{prop}

Constructing coequalisers is slightly more complicated. Recall
that given two ordinary functions $\begin{tikzcd} A \arrow[r, "f",
shift left] \arrow[r, "g", swap, shift right] & B \end{tikzcd}$,
their coequaliser is the set $B$ quotiented by the equivalence
relation $\sim_{f, g}$, which is the least equivalence relation
such that $(f(a), g(a)) \in \mathop{\sim_{f, g}}$.  The elements
of $B/\mathop{\sim_{f, g}}$ are then equivalence classes $[b]$ of
elements we wish to `lump together.' Suppose now that $A$ and $B$
are classified; how should we classify $B/\mathop{\sim_{f, g}}$?
We may consider two of its equivalence classes indistinguishable
whenever it happens that two elements, one from each equivalence
class, are `lumped together' by $\rel{\ell}$. So we define \[
  [b] \rel{\ell} [b']
    \Longleftrightarrow
  \exists x \in [b].\ \exists y \in [b'].\ x \rel{\ell} y
\] and thus turn $B/\mathop{\sim_{f, g}}$ into a classified set.
The quotient map $B \rightarrow B/\mathop{\sim_{f, g}}$ is then
automatically a morphism of classified sets, and it is not hard to
show that 

\begin{prop} 
  $B/{\sim_{f, g}}$ is the coequaliser of $\begin{tikzcd} A
  \arrow[r, "f", shift left] \arrow[r, "g", swap, shift right] & B
  \end{tikzcd}$ $\mathbf{CSet}_\mathcal{L}$.
\end{prop}

\noindent In short, we have the following theorem:

\begin{thm}
  $\mathbf{CSet}_\mathcal{L}$ is finitely complete and finitely
  cocomplete.
\end{thm}

\subsection{Exponentials}

In a manner identical to that of logical relations, we are able
endow the set of morphisms from a classified set to another with
an indistinguishability relation. The idea is that two morphisms
are indistinguishable if they map indistinguishable inputs to
indistinguishable outputs. Note that this furnishes our theory
with an \emph{extensional} view of functions, where they are
understood to be indistinguishable precisely when their
input-output behaviour is.

Given classified sets $A$ and $B$ over $\mathcal{L}$, we define
the classified set $B^A$ by \begin{align*}
  \bars{B^A} &\myeq \text{Hom}_{\textbf{CSet}_\mathcal{L}}(A, B) \\
  f \rel{\ell} g &\Longleftrightarrow \forall a \rel{\ell} a'.\
  f(a) \rel{\ell} g(a')
\end{align*}

\noindent This is the usual definition of logical relations at
function types. We can then define a function $\textsf{ev} : B^A
\times A \rightarrow B$ by $(f, a) \mapsto f(a)$, and the
definition of $B^A$ makes it a morphism. From that point, it is
trivial to show that

\begin{prop}
  $B^A$ is the exponential of $A$ and $B$.
\end{prop}

\noindent Hence,

\begin{thm}
  $\mathbf{CSet}_\mathcal{L}$ is a bicartesian closed category.
\end{thm}

Thus the category $\mathbf{CSet}_\mathcal{L}$ is rich enough to
model the simply typed $\lambda$-calculus with coproducts. In the
following sections we will also show that there is enough
structure to model a multitude of \emph{modal operators} on types.

\section{Cohesion}
  \label{sec:coh}

The modal structure on classified sets is closely related to---or,
more precisely, induced by---Lawvere's \emph{axiomatic cohesion}.
But what is axiomatic cohesion? It is a theory developed by
\citet{Lawvere2007} as an attempt to capture the very broad idea
of \emph{mathematical spaces} that are endowed with some kind of
\emph{cohesion}, i.e. the idea that some points are `very close to
each other' or `stuck together.' The prototypical example is that
of \emph{topological spaces}, which are sets of points equipped
with a \emph{topology}, i.e. a set of subsets of these points.
These are called the \emph{open sets}, and the choice of which
subsets are open endows the underlying set with notions of
continuity, connectedness, convergence etc. 

If we write $\textbf{Top}$ for the category of topological spaces,
there is an obvious forgetful functor $U : \textbf{Top}
\longrightarrow \textbf{Set}$ which `forgets' the cohesive
structure of a topological space, and returns the underlying set
of points. Conversely, given any set there are two canonical ways
to construct a topological space. The first is to specify that
\emph{no points are stuck together}. This is achieved by the
\emph{discrete topology}, where every subset is an open set. It
forms a functor, \[
  \Updelta : \textbf{Set} \longrightarrow \textbf{Top}
\] that is the identity on functions: every function on sets is
trivially a continuous function between the same sets seen as
discrete spaces, so $\Updelta$ is full and faithful.

The other way is to specify that \emph{all points are stuck
together}, and is achieved by endowing the set with the
\emph{codiscrete topology}, where the only open sets are the empty
set and the entire space. It forms another functor, \[
  \nabla : \textbf{Set} \longrightarrow \textbf{Top}
\] that is the identity on functions: every function on sets is
trivially a continuous function between the same sets seen as
codiscrete spaces.

The relationship between these functors is simple: they form a
\emph{string of adjoints}: \[
  \Updelta \dashv U \dashv \nabla
\] This categorifies the following two simple observations: every
function $X \rightarrow \nabla Y$ from a topological space into a
codiscrete space is continuous (as everything is collapsed into a
single block of points); and every function $\Updelta X
\rightarrow Y$ from a discrete space into a topological space $Y$
is continuous (as there is no cohesion to preserve in a discrete
space).

But this is not the whole story. Consider a continuous function $X
\rightarrow \Updelta Y$ into the discrete space with points $Y$.
Since it is continuous, it must preserve cohesion. Namely, it must
map points of $X$ that are `stuck together' to points that are
`stuck together' in $\Updelta Y$. But points are only `stuck' to
themselves in the discrete space $\Updelta Y$, so in fact it must
map all points `stuck' to each other in $X$ to a single point of
$Y$. Thus, if we could somehow reduce $X$ down to a set $C(X)$
where points `stuck' together are collapsed to a single point, the
continuous function $X \rightarrow \Updelta Y$ would define an
ordinary function $C(X) \rightarrow Y$. Such a functor $C$
\emph{does} exist, and maps the topological space $X$ to its set
$C(X)$ of \emph{connected components}. It is evidently left
adjoint to $\Updelta$. 

We are very close to showing that topological spaces are
\emph{cohesive relative to sets}. We will, however, not use this
full notion in this paper, as the weaker notion of
\emph{pre-cohesion}, due to \citet{Lawvere2015}, is more than
sufficient for our purposes:

\begin{defn}
  In a situation of the form \[
    \begin{tikzcd}
      \mathcal{E}
        \arrow[d, shift right=45pt, "C"{swap}, ""{name=A}]
        \arrow[d, shift left=15pt, "U", ""{name=B}]
      \\
      \mathcal{S}
        \arrow[u, shift left=15pt, "\Updelta"{swap}, ""{name=C}]
        \arrow[u, shift right=45pt, "\nabla"{swap}, ""{name=D}]
          \arrow[from=A, to=C, symbol=\dashv]
          \arrow[from=C, to=B, symbol=\dashv]
          \arrow[from=B, to=D, symbol=\dashv]
    \end{tikzcd}
  \] where $\mathcal{E}$ and $\mathcal{S}$ are extensive
  categories, we call \emph{$\mathcal{E}$ pre-cohesive relative to
  $\mathcal{S}$} if 
    \begin{enumerate}
      \item
        $\Updelta, \nabla : \mathcal{S} \longrightarrow \mathcal{E}$
        are full and faithful;
      \item
        $C : \mathcal{E} \longrightarrow \mathcal{S}$ preserves finite
        products; and
      \item
        the \emph{Nullstellensatz} holds: the counit
        $\textsf{Id}_\mathcal{E} \Rightarrow \Updelta C$ is an
        epimorphism; or, equivalently, the unit $\Updelta U
        \Rightarrow \textsf{Id}_\mathcal{E}$ is a monomorphism.
    \end{enumerate}
\end{defn}

The second requirement essentially expresses that a connected
component of a product space is exactly a connected component in
each of the two components. The third requirement, the
\emph{nullstellensatz}, has many equivalent forms, and essentially
requires that the `quotient' map that maps a point to its
connected component is a epimorphic, i.e. a kind of abstract
surjection (in other words: no connected component is empty).

The reason that a setting of pre-cohesion is of direct interest to
modal type theory is that it automatically induces \emph{three
modalities} on the category $\mathcal{E}$ by composing each pair
of adjoints. The first one, \[
  \Box \myeq \Updelta U : \mathcal{E} \longrightarrow \mathcal{E}
\] is a comonad. Intuitively, $\Box$ takes a cohesive space,
strips its points of their cohesive structure, and gives them the
discrete structure: it `unsticks' all points. The second one, \[
  \blacklozenge \myeq \nabla U : \mathcal{E} \longrightarrow
  \mathcal{E}
\] is a monad, which does the opposite: it `sticks' all the points
of a cohesive space together. Finally, \[
  \int \myeq \Updelta C : \mathcal{E} \longrightarrow \mathcal{E}
\] is a monad. Intuitively, $\int$ collapses each connected
component into a single point, and then presents that set of
connected components as a discrete space.

Summarising, the above setting endows these functors with the
following properties: 
\begin{cor}[Fundamental Corollary of Pre-Cohesion] \hfill
  \label{cor:coh}
  \begin{enumerate}
    \item
      $U : \mathcal{E} \longrightarrow \mathcal{S}$ preserves
      limits and colimits.
    \item
      $\Updelta : \mathcal{S} \longrightarrow \mathcal{E}$
      preserves limits and colimits.
    \item
      $\nabla : \mathcal{S} \longrightarrow \mathcal{E}$ preserves
      limits.
    \item
      $C : \mathcal{E} \longrightarrow \mathcal{S}$ preserves
      products and colimits.
    \item
      $U\Updelta \cong \textsf{Id}_\mathcal{S} : \mathcal{S}
      \longrightarrow \mathcal{S}$
    \item
      $U\nabla \cong \textsf{Id}_\mathcal{S} : \mathcal{S}
      \longrightarrow \mathcal{S}$
    \item
      $\Box \myeq \Updelta U : \mathcal{E} \longrightarrow
      \mathcal{E}$ is an idempotent comonad. It is \emph{exact},
      i.e. preserves finite limits and colimits.
    \item
      $\blacklozenge \myeq \nabla U : \mathcal{E} \longrightarrow
      \mathcal{E}$ is an idempotent monad. It is \emph{left
      exact}, i.e. preserves finite limits.
    \item
      $\int \myeq \Updelta C : \mathcal{E} \longrightarrow
      \mathcal{E}$ is an idempotent monad. It preserves products
      and colimits.
    \item
      $\int \dashv \Box \dashv \blacklozenge$
  \end{enumerate}
\end{cor}
\begin{proof}
  (1)-(4) follow by the fact each of these functors is a left or
  right adjoint, or by some assumption. (5) and (6) follow by the
  Yoneda lemma and the fact $\Updelta$/$\nabla$ are f.f.; e.g. for
  any $D \in \mathcal{S}$ \[
    \text{Hom}_\mathcal{S}(D, U\Updelta A)
      \cong
    \text{Hom}_\mathcal{E}(\Updelta D, \Updelta A)
      \cong
    \text{Hom}_\mathcal{S}(D, A)
  \] (7)-(9) follow from (1)-(4) and the fact $\Updelta$/$\nabla$
  are full and faithful and thus generate idempotent (co)monads
  (see e.g. \cite[\S 4.3.2]{Borceux1994}). (10) follows from $C
  \dashv \Updelta \dashv U \dashv \nabla$.
\end{proof}

The rest of the section is devoted to showing that 

\begin{thm}
  \label{thm:coh}
  $\mathbf{CSet}_\mathcal{L}$ is pre-cohesive relative to
  $\mathbf{Set}$.
\end{thm}

A variant of this fact (namely that reversible reflexive graphs
are cohesive relative to sets) is already present in
\cite{Lawvere2007}. There is an evident forgetful functor from
classified sets to sets: \begin{align*}
  U : \textbf{CSet}_\mathcal{L}&\longrightarrow \textbf{Set} \\
      X             &\longmapsto \bars{X}
\end{align*} which forget all the relations $\rel{\ell}$. We can
then return to classified sets using the functor \[
  \Updelta : \textbf{Set} \longrightarrow
  \textbf{CSet}_\mathcal{L}
\] which adds to the set $X$ the diagonal relation $x \rel{\ell}
x$ at each level $\ell \in \mathcal{L}$. This is the \emph{finest}
equality that can be supported by this setting, in that each
element is only indistinguishable to itself. In that sense,
$\Updelta X$ is a classified set with carrier $X$ that is
\emph{completely transparent}, in a manner reminiscent of the
discrete topology on a set. Any function $f : X \rightarrow Y$ is
trivially a function $f : \Updelta X \rightarrow \Updelta Y$, as
the diagonal relation is trivially preserved; thus $\Updelta$ is
indeed a functor. Moreover, it is easy to see that it is full and
faithful, and that

\begin{prop}
  $\Updelta : \textbf{Set} \longrightarrow
  \textbf{CSet}_\mathcal{L}$ is left adjoint to the forgetful
  functor.
\end{prop}
\begin{proof}
  Any morphism $f : \Updelta X \rightarrow Y$ is a function $f : X
  \rightarrow \bars{Y} = UY$. Conversely, any function $f : X
  \rightarrow UY$ can be seen as a morphism $f : \Updelta X
  \rightarrow Y$ of classified sets, as it trivially preserves the
  diagonal relation. Naturality is trivial.
\end{proof}

We then define \[
  \nabla : \textbf{Set} \longrightarrow \textbf{CSet}_\mathcal{L}
\] to map a set $X$ to itself, but equipped with the
\emph{complete relation} $\rel{\ell} \myeq \bars{X} \times
\bars{X}$ at each $\ell \in \mathcal{L}$. That is: no element of
$\nabla X$ is distinguishable from any other. Thus $\nabla X$ is
the classified set with carrier $X$ that is \emph{maximally
opaque}. This reminds us of the codiscrete topology on $X$.
Again, rather trivially, any function $f : X \rightarrow Y$ can be
seen as a morphism $f : \nabla X \rightarrow \nabla Y$, as it
evidently preserves the complete relation on $X$. It is not hard
to see that $\nabla$ is a full and faithful functor, and that

\begin{prop}
  $\nabla : \textbf{Set} \longrightarrow
  \textbf{CSet}_\mathcal{L}$ is right adjoint to the forgetful
  functor.
\end{prop}

\begin{proof}
  Any function $f : \bars{X} \rightarrow Y$ can be seen as a
  morphism $f : X \rightarrow \nabla Y$: it trivially preserves
  all the related pairs in $X$---no matter what they are---for
  $\nabla Y$ relates all elements of $Y$. Conversely, any
  morphism $f : X \rightarrow \nabla Y$ is simply a function $f :
  \bars{X} \rightarrow \bars{\nabla Y} = Y$. Naturality is again
  trivial.
\end{proof}

We now move on to connected components. Suppose we have a morphism
$f : X \rightarrow \Updelta Y$. Then, as each $\rel{\ell}$ in
$\Updelta Y$ is simply reflexivity, we have that for any $\ell \in
\mathcal{L}$, \[
  x \rel{\ell} x' \text{ (in $X$)}\
    \Longrightarrow
  f(x) \rel{\ell} f(x') \text{ (in $\Updelta Y$)}\
    \Longrightarrow
  f(x) = f(x')
\] That is: $f$ collapses related elements of $X$ that are related
at some---any!---level to a single element in $Y$. We cannot
phrase this in terms of quotients yet, for $\rel{\ell}$ need not
be an equivalence relation. So, let us define the relation
$R^\star \subseteq \bars{X} \times \bars{X}$ to be the reflexive,
symmetric, transitive closure of $\bigcup_{\ell \in \mathcal{L}}
\rel{\ell}$, i.e. the \emph{least} equivalence relation containing
all the $R_\ell$'s. We can now define $\bars{X}/R^\star$. This
extends to a functor \[
  C : \textbf{CSet}_\mathcal{L} \longrightarrow \textbf{Set}
\] by letting $C(X) \myeq \bars{X} / R^\star$, and defining $Cf
\myeq f^\star : \bars{X} / R^\star \rightarrow \bars{Y} /
R^\star$, where \[
  f^\star([x]) \myeq [f(x)]
\] $f^\star$ is well-defined: if $x \mathrel{R^\star} x'$, then
there is a (possibly empty) sequence $x_0, \dots, x_n$ of elements
of $X$ and a sequence of levels $\ell_0, \dots, \ell_{n+1}$ such
that \[
  x \rel{\ell_0} x_0 \rel{\ell_1}^{-1} x_1 \dots
                     \rel{\ell_{n-1}}^{-1} x_{n-1} 
                     \rel{\ell_n} x_n
                     \rel{\ell_{n+1}} x'
\] But $f$ preserves all of these, mapping inverses to inverses,
so \[
  f(x) \rel{\ell_0} f(x_0) \rel{\ell_1}^{-1} f(x_1) \dots
                   \rel{\ell_{n-1}}^{-1} f(x_{n-1})
                   \rel{\ell_n}^{-1} f(x_n) \rel{\ell_{n+1}} f(x')
\] and hence $[f(x)] = [f(x')]$, so the choice of equivalence
class representative is immaterial. We hence obtain a functor that
gives us for each classified set $X$ its set of connected
components, i.e. the set of equivalence classes
$\bars{X}/R^\star$. Any two elements in the same equivalence class
are indistinguishable at some level $\ell \in \mathcal{L}$.

Let us return to the function $f : X \rightarrow \Updelta Y$.
Define the relation \[
  x \sim_f x'
    \quad\myeq\quad
  f(x) = f(x')
\] This is clearly an equivalence relation, and we know that
$\bigcup_{\ell \in \mathcal{L}} \rel{\ell}\ \subseteq\
\mathrel{\sim_f}$.  But $R^\star$ is the least equivalence
relation such that the above is true, so $R^\star \subseteq\
\mathrel{\sim_f}$.  Hence, $Uf : \bars{X} \rightarrow Y$, which
respects $\bigcup_{\ell} \rel{\ell}$, can be uniquely factored as
\[
  \begin{tikzcd}
    \bars{X}
      \arrow[r, "\eta_X"]
      \arrow[dr, "Uf", swap]
    & \bars{X} / R^\star 
      \arrow[d, "\hat{f}", dotted] \\
    & Y
  \end{tikzcd}
\] where $\eta_X$ is the quotient map, and $\hat{f}([x]) \myeq
f(x)$. So we can uniquely take $f : X \rightarrow \Updelta Y$ to
$\hat{f} : CX \rightarrow Y$, and

\begin{prop}
  $C : \mathbf{CSet}_\mathcal{L} \rightarrow \mathbf{Set}$ is left
  adjoint to $\Updelta : \mathbf{Set} \rightarrow
  \mathbf{CSet}_\mathcal{L}$.
\end{prop}

It is easy to see that the canonical map $C(X \times Y)
\rightarrow CX \times CY$ is $[(x, y)] \mapsto \left([x],
[y]\right)$, and that, due to the behaviour of the relations on
the product, it has an inverse. Moreover, in the case of
classified sets, the \emph{nullstellensatz} is trivial: applying
the hom-set isomorphism of $\Updelta \dashv U$ to the identity
$id_{UX} : UX \rightarrow UX$, which is really the identity
function $id_{\bars{X}} : \bars{X} \rightarrow \bars{X}$, yields
the unit $\Updelta U X \rightarrow X$, which is again the identity
function on $\Updelta U X = \Updelta \bars{X}$, which is injective
and hence monic.

Thus, Theorem \ref{thm:coh} holds.

\section{Noninterference I: Monads and Comonads}
  \label{sec:noninterference1}

The cohesive structure of classified sets that we have developed
so far is enough to show some basic information flow properties
for monadic and comonadic calculi. In this section, we will state
and prove noninterference properties for (a) Moggi's monadic
metalanguage, and (b) the Davies-Pfenning calculus. Both of these
properties follow from an axiom of cohesion, which we call
\emph{codiscrete contractibility}, and which we discuss first.

Central to our noninterference proofs will be two canonically
constructed objects; they will be classified sets over the carrier
\[
  \mathbb{B} \myeq \{\texttt{tt}, \texttt{ff}\}
\] of booleans. The first one is the `discrete' booleans $\Updelta
\mathbb{B}$, in which if $b \rel{\ell} b'$ it follows that $b =
b'$: these are the booleans that are visible to
everyone.\footnote{They are referred to as the `low security
booleans' by \citet{Abadi1999}, and denoted \emph{boolL}.} The
second one will be the `codiscrete' booleans $\nabla \mathbb{B}$,
in which $b \rel{\ell} b'$ for all pairs of booleans $(b, b')$,
which are invisible at all levels. Note that \begin{align*}
  \mathbb{B} &\cong \mathbf{1} + \mathbf{1} \\ 
  \Updelta\mathbb{B} &\cong \Updelta\mathbf{1} + \Updelta\mathbf{1}
  \cong \mathbf{1} + \mathbf{1} \\
  \nabla\mathbb{B} 
  &\cong \nabla U \Updelta \mathbb{B}
    \cong \blacklozenge(\mathbf{1} + \mathbf{1})
\end{align*} These all follow from the fundamental corollary
(Corollary \ref{cor:coh}). The first isomorphism is by definition.
The second holds because $\Updelta$ preserves isomorphisms,
colimits ($+$) and limits $(\mathbf{1})$. The final one is
obtained by using $U\Updelta \cong \textsf{Id}$, then the second
one, and then the definition $\blacklozenge \myeq \nabla U$.

\subsection{Contractible codiscreteness and information flow}

We are now ready to discuss a certain axiom that pre-cohesions may
or may not satisfy. This axiom expresses a very intuitive
property: it says that when all points are stuck together they
constitute at most one connected component. Surprisingly, this
single axiom is the main point of connection between cohesion and
information flow: we will use it to prove results which are at the
core of our noninterference theorems, both here and in
\S\ref{sec:noninterference2}.

We have carefully described the string of adjoints $C \dashv
\Updelta \dashv U \dashv \nabla$ that expresses the idea of
pre-cohesion. In particular, $\nabla X$ is intuitively understood
to be the object with point-set $X$ equipped with maximal
cohesion, i.e. `everything stuck together.' A property that is
should follow from this intuition is that if one were to look at
the connected components of $\nabla X$, one would find \emph{at
most one}---or none, if $X$ is the empty set. Unfortunately, this
is not something that readily follows from pre-cohesion, but we
can explicitly ask for it.

\begin{defn}
  If $\mathcal{E}$ is pre-cohesive relative to $\mathcal{S}$ with
  $C \dashv \Updelta \dashv U \dashv \nabla$, then this
  pre-cohesion satisfies \emph{contractible codiscreteness} if
  codiscrete objects have at most a single connected component.
  That is: if for any $X \in \mathcal{S}$, $C(\nabla X)$ is a
  subobject of the terminal object $\mathbf{1}$, i.e. the unique
  morphism \[
    C(\nabla X) \rightarrow \mathbf{1}
  \] is monic.
\end{defn}

\noindent \citet{Lawvere2015} consider this concept in toposes,
and call it \emph{connected codiscreteness}.  The name for our
more general setting is due to \citet{Shulman2018}.

It is easy to see that 
\begin{prop}
  The pre-cohesion of $\textbf{CSet}_\mathcal{L}$ relative to
  $\textbf{Set}$ satisfies contractible codiscreteness.
\end{prop}

\noindent This is a fancy way of stating the following obvious
fact: if we start with a set $X$, equip it with the complete
relation at each level $\ell \in \mathcal{L}$, and we then take
the connected components of that, there will be at most one. In
fact, if $X = \emptyset$ there will be none, and otherwise there
will be exactly one.

The axiom of contractible codiscreteness is very useful, as it
allows one to prove \emph{abstract noninterference properties}. We
will prove such a property presently, but we first have to discuss
its slightly thorny interaction with non-emptiness in categorical
terms.

Suppose the object $X$ is \emph{non-empty}, i.e. there is a point
$x : \mathbf{1} \rightarrow X$. This is a sufficient condition to
show that $C(\nabla X)$ is actually `contractible,' i.e. isomorphic
to the terminal object.

\begin{prop}[Contractibility of non-empty codiscrete spaces]
  \label{prop:cc}
  Let $\mathcal{E}$ be pre-cohesive over $\mathcal{S}$ in a way
  that satisfies contractible codiscreteness. If $X \in
  \mathcal{S}$ is non-empty, then \[
    C(\nabla X) \cong \mathbf{1}
  \]
\end{prop}
\begin{proof}
  Let $x : \mathbf{1} \rightarrow X$. By applying $C\nabla$ to it,
  and then using preservation of products (which we have by Lemma
  \ref{cor:coh}), we obtain a point $\mathbf{1} \rightarrow
  C\nabla X$, which is to say that $C\nabla X$ is non-empty too.
  The composite $\mathbf{1} \rightarrow C\nabla X \rightarrow
  \mathbf{1}$ is trivially the identity. Moreover, the composite
  $C\nabla X \rightarrow \mathbf{1} \rightarrow C\nabla X$ is the
  identity on $C\nabla X$: as $C\nabla X \rightarrow \mathbf{1}$
  is monic, any two morphisms into $C\nabla X$ are equal.
\end{proof}

\noindent We can now state and prove the following.

\begin{prop}
  \label{prop:cocoint}
  Let $\mathcal{E}$ be pre-cohesive over $\mathcal{S}$, in a way
  that satisfies contractible codiscreteness. Then
  \begin{enumerate}
    \item
      Morphisms $\blacklozenge A \rightarrow \Updelta B$ for
      non-empty $A \in \mathcal{E}$ and $B \in \mathcal{S}$
      correspond to points $\mathbf{1} \rightarrow B$.
    \item 
      Morphisms $\nabla A \rightarrow \Box B$ for non-empty $A \in
      \mathcal{S}$ and $B \in \mathcal{E}$ correspond to points
      $\mathbf{1} \rightarrow UB$.
  \end{enumerate}
\end{prop}
\begin{proof} \hfill
  \begin{enumerate}
    \item
      The following isomorphisms hold naturally: \begin{align*}
        \mathcal{E}(\blacklozenge A, \Updelta B)
          &\cong \mathcal{E}(\nabla U A, \Updelta B) 
          & \text{ by definition of $\blacklozenge$} \\
          &\cong \mathcal{S}(C \nabla U A, B) 
          & \text{ as $C \dashv \Updelta$} \\
          &\cong \mathcal{S}(\mathbf{1}, B) 
          & \text{ by Prop. \ref{prop:cc}}
      \end{align*}
    \item
      The following isomorphisms hold naturally: \begin{align*}
        \mathcal{E}(\nabla A, \Box B)
          &\cong \mathcal{E}(\nabla A, \Updelta U B) 
          & \text{ by definition of $\Box$} \\
          &\cong \mathcal{S}(C \nabla A, U B) 
          & \text{ as $C \dashv \Updelta$} \\
          &\cong \mathcal{S}(\mathbf{1}, U B) 
          & \text{ by Prop. \ref{prop:cc}}
      \end{align*}
  \end{enumerate}
\end{proof}

In the concrete case of $\textbf{CSet}_\mathcal{L}$ the above
lemma says that if $A$ is non-empty then morphisms of type
$\blacklozenge A \rightarrow \Updelta B$ and $\Updelta A
\rightarrow \Box B$ are \emph{constant functions}. Indeed, the
first isomorphism in either proof is an identity, and the second
one collapses all elements of $\nabla U A$ or $\nabla A$ to a
single element.

A special case of (1), `manually' proven for the particular case
of $\textbf{CSet}_\mathcal{L}$ and $B = \mathbb{B}$, forms the
central reasoning involved in the noninterference proofs of
\citet{Abadi1999}. Our result generalises this to any setting of
pre-cohesion, and also yields the heretofore unnoticed dual (2).

We will now apply (1) and (2) to construct two noninterference
proofs.

\subsection{Moggi's monadic metalanguage}
  \label{sec:moggi}

\citet{Moggi1991} introduced the \emph{monadic metalanguage}, a
typed $\lambda$-calculus which, for each type $A$ features a type
of \emph{computations} $TA$. The `result' of a computation $M :
TA$ is a `value' of type $A$. Indeed, each `value' is a trivial
computation, as can be witnessed by the introduction rule: \[
  \begin{prooftree}
      \Gamma \vdash M : A
    \justifies
      \Gamma \vdash [M] : TA
  \end{prooftree}
\] The idea is that types of the form $TA$ encapsulate various
`notions of computations,' nowadays referred to as \emph{effects}
(recursion, nondeterminism, operations on the store, etc.). Very
little is assumed of the type constructor. To quote Moggi:
\begin{quote}
  ``Rather than focusing on a specific $T$, we want to find the
  general properties common to all notions of computation.''
\end{quote}

Moggi then showed in \cite[\S 3]{Moggi1991} that there is a
categorical semantics of this language in which $T$ is interpreted
by a \emph{strong monad} on a CCC $\mathcal{C}$. The category
$\mathcal{C}$ is then the \emph{category of values}, whereas the
Kleisli category $\mathcal{C}_T$ of the monad $T : \mathcal{C}
\longrightarrow \mathcal{C}$ is the \emph{category of programs}.

The key information flow property enjoyed by monads is that once
something is `inserted' into the monad, then it cannot flow out.
The monad identifies a region of the language which is
\emph{impure}, in that evaluating terms within the region causes
\emph{effects}. If the results of those effects were to flow
`outside' the monad, the outcome would be the loss of referential
transparency. This information flow property is evident if one
looks at the elimination rule for $T$: \[
  \begin{prooftree}
      \Gamma \vdash M : TA
    \quad
      \Gamma , x : A \vdash N : TB
    \justifies
      \Gamma \vdash \textsf{let } x = M \textsf{ in } N : TB
  \end{prooftree}
\] If we have a computation $M : TA$ that yields a result of type
$A$, the $\textsf{let}$ construct allows us to substitute it for a
variable of type $A$, but only as long as this to be used within
\emph{another} computation $N : TB$. Thus values of type $TA$
cannot `escape' the scope of $T$.

It is not hard to show that $\mathbf{CSet}_\mathcal{L}$ is a model
of Moggi's monadic metalanguage (with booleans), with $T$ being
interpreted by $\blacklozenge$, and $\mathsf{Bool}$ by
$\Updelta\mathbb{B} \cong \mathbf{1} + \mathbf{1}$: the only thing
that remains to be shown is that $\blacklozenge$ is a strong
monad, which is true by a more general result that we cover later
(Prop. \ref{prop:strong}). If we use the standard categorical
interpretation, as defined by \cite[\S 3]{Moggi1991}, a term $x :
TA \vdash M : \textsf{Bool}$ is interpreted as a morphism
$\blacklozenge\sem{A}{} \rightarrow \Updelta\mathbb{B}$ in
$\mathbf{CSet}_\mathcal{L}$. Hence, by Prop. \ref{prop:cocoint} we
know it is interpreted by a constant function in the model.

Nevertheless, this is not a noninterference result yet, for it
tells us something about a particular model, viz.
$\mathbf{CSet}_\mathcal{L}$, but nothing about the equational
theory of Moggi's calculus. To infer something about it from one
of its models, some \emph{completeness property} of that model
must be established. We will use the simplest one of all, which is
known as \emph{adequacy}. It is the most basic form of
completeness, and is commonly used in the study of calculi with
recursion, e.g. PCF \cite{Plotkin1977, Streicher2006}. A model of
a calculus is adequate exactly when it is complete at \emph{ground
types}. What we take to be a ground type is up to us, but the
usual choices are simple base types (booleans, flat naturals,
etc.) For this paper, we use the term to refer to types defined by
a collection of constants, along with elimination and computation
rules. For example, $\textsf{Bool}$ is a ground type: it may be
presented by the constants \[
  \Gamma \vdash \texttt{tt} : \textsf{Bool}
    \qquad
  \Gamma \vdash \texttt{ff} : \textsf{Bool}
\] along with the the elimination rule \[
  \begin{prooftree}
      \Gamma \vdash M : \textsf{Bool}
        \quad
      \Gamma \vdash E_0 : C
        \quad
      \Gamma \vdash E_1 : C
    \justifies
      \Gamma \vdash \textsf{if } M \textsf{ then } E_1 \textsf{
      else } E_0 : C
  \end{prooftree}
\] and the two computation rules \begin{align*}
  \textsf{if } \texttt{tt} \textsf{ then } E_1 \textsf{ else } E_0 &= E_1 \\
  \textsf{if } \texttt{ff} \textsf{ then } E_1 \textsf{ else } E_0 &= E_0 
\end{align*} 

Suppose that a (non-dependent, possibly modal) type theory
satisfies \emph{canonicity at all ground types}, i.e. for every
such ground type $G$ and every \emph{closed} term $\vdash M : G$
there is a unique constant $\textsf{c} : G$ such that $\vdash M =
c : G$. For example, canonicity for $\textsf{Bool}$ requires that
every closed term $\vdash M : \textsf{Bool}$ is equal to either
$\texttt{tt}$ or $\texttt{ff}$. In strongly normalising
programming calculi, canonicity is a corollary of \emph{progress
and preservation theorems} \citep[\S8.3, \S9]{Pierce2002}: the
normalisation of a well-typed closed term will reach a
\emph{canonical form}. In type theories, canonicity is a corollary
of \emph{confluence and strong normalisation}: each closed term
can be reduced to a unique normal form, and each closed normal
form of type $A$ can only correspond an introduction rules of type
$A$, which for ground types are simply constants.

Suppose now that we interpret the type theory in a category in the
standard way---e.g. as in \cite{Crole1993, Abramsky2011a}---so
that (a) closed terms $\vdash M : G$ are interpreted as points
$\mathbf{1} \rightarrow \sem{G}{}$, (b) the interpretation is
sound, in that $\vdash M = N : A$ implies $\sem{M}{} = \sem{N}{}$,
and (c) the interpretation of each ground type $G$ is
\emph{injective}, in that distinct constants have distinct
interpretations. Then, the interpretation is automatically
adequate:

\begin{lem}[Adequacy]
  \label{lem:adeq}
  Suppose that a type theory satisfies canonicity at ground types,
  and has a sound categorical interpretation which is injective at
  every ground type $G$, in the sense that \[
    \sem{\vdash \textsf{c}_i : G}{}
      =
    \sem{\vdash \textsf{c}_j : G}{}
      : \mathbf{1} \rightarrow \sem{G}{}
      \quad\Longrightarrow\quad
    \textsf{c}_i \equiv \textsf{c}_j
  \] Then this interpretation is \emph{adequate for $G$}, in the
  sense that \[
    \sem{\vdash M : G}{} = \sem{\vdash \textsf{c}_i : G}{}
      \quad\Longrightarrow\quad
    \vdash M = \textsf{c}_i : G
  \]
\end{lem}
\begin{proof}
  By canonicity, for any $\vdash M : G$ we have $\vdash M = c_j :
  G$ for some constant $c_j$ of $G$. But then
  $\sem{\textsf{c}_i}{} = \sem{M}{} = \sem{\textsf{c}_j}{}$ by
  soundness, so $\textsf{c}_i \equiv \textsf{c}_j$, and hence
  $\vdash M = \textsf{c}_i : G$.
\end{proof}

\noindent This lemma applies to the standard interpretation of
Moggi's monadic metalanguage into any CCC with a strong monad.  It
is known that the straightforward extension of Moggi's calculus
with coproducts and a unit type (which together subsume
$\textsf{Bool}$) is confluent, strongly normalising, and has a
sound interpretation into any biCCC with a strong monad: this was
shown by \citet{Benton1998}. Thus, it satisfies canonicity at
$\textsf{Bool}$. Hence, the above lemma still applies, and we
obtain

\begin{thm}[Noninterference for Moggi]
  Let $A$ be a non-empty type, which is to say there exists a
  closed term of type $A$. If $x : TA \vdash M : \textsf{Bool}$
  then for any $\vdash E, E' : TA$ we have \[
    \vdash M[E/x] = M[E'/x] : \textsf{Bool}
  \]
\end{thm}
\begin{proof}
  The interpretation of $\textsf{Bool}$ as $\Updelta\mathbb{B}
  \cong \mathbf{1} + \mathbf{1}$ in $\mathbf{CSet}_\mathcal{L}$
  satisfies the assumptions of Lemma \ref{lem:adeq}. Hence, the
  interpretation is adequate for it. We have that \[
    \sem{ \vdash M[F/x] : \textsf{Bool}}{} 
    = \sem{x : TA \vdash M : \textsf{Bool}}{}
        \circ \sem{\vdash F : TA}{} : 
          \mathbf{1} \rightarrow \Updelta\mathbb{B}
  \] for any $\vdash F : TA$. But as $A$ is non-empty, we have
  that $\sem{A}{}$ is non-empty, so by Proposition
  \ref{prop:cocoint}(1), $\sem{x : TA \vdash M : \textsf{Bool}}{}
  : \blacklozenge\sem{A}{} \rightarrow \Updelta\mathbb{B}$ is a
  constant function, so \[
    \sem{M[E/x]}{} 
      = \sem{x : TA \vdash M : \textsf{Bool}}{} \circ \sem{E}{}
      = \sem{x : TA \vdash M : \textsf{Bool}}{} \circ \sem{E'}{}
      = \sem{M[E'/x]}{}
  \] for any $\vdash E, E' : TA$. By adequacy, it follows that
  $\vdash M[E/x] = M[E'/x] : \textsf{Bool}$.
\end{proof}

\subsection{The Davies-Pfenning calculus}
  \label{sec:dp}

Davies and Pfenning introduced a comonadic modal type theory as a
type system for binding-time analysis in \citep{Davies1996,
Davies2001a}. The idea was that data could not arbitrarily flow
from type $A$ to type $\Box A$. This is clearly reflected in the
structure of the type system: each typing judgement comes with two
contexts, and has the shape \[
  \ctxt{\Updelta}{\Gamma} \vdash M : A
\] where $\Updelta$ are the `modal' variables, and $\Gamma$ are
normal variables. A `modal' variable can be used as a normal
variable, as witnessed by the rule $\ctxt{\Updelta, u : A}{\Gamma}
\vdash u : A$. However, when introducing terms of type $\Box A$,
we can only use `modal' variables from $\Updelta$: \[
  \begin{prooftree}
      \ctxt{\Updelta}{\cdot} \vdash M : A
    \justifies
      \ctxt{\Updelta}{\Gamma} \vdash \ibox{M} : \Box A
  \end{prooftree}
\] It is not hard to define a non-trivial morphism $\textsf{Bool}
\rightarrow \Box \textsf{Bool}$: \[
  b : \textsf{Bool} \vdash \textsf{ if } b \textsf{ then }
  \ibox{\texttt{tt}} \textsf{ else } \ibox{\texttt{ff}} :
  \Box\textsf{Bool}
\] However, it is impossible to pass from $A \rightarrow B$ to
$\Box(A \rightarrow B)$ in general. This property was used by
Davies and Pfenning to separate things that were available
\emph{statically}, by isolating them under the box, from things
that are available \emph{dynamically} (and thus cannot always be
used for metaprogramming). For example, one can always use a
boolean constant for metaprogramming, as shown above, but one
cannot use `live' piece of code, a function $A \rightarrow B$.
However, if one has prudently arranged to have a copy of the
source code of a function $A \rightarrow B$ available, that would
be of type $\Box(A \rightarrow B)$, and one can then use it for
metaprogramming.

Terms of type $\Box A$ are used by substituting for a variable in
the `modal' context. This yields the following elimination rule:
\[
  \begin{prooftree}
      \ctxt{\Updelta}{\Gamma} \vdash M : \Box A
        \quad
      \ctxt{\Updelta, u : A}{\Gamma} \vdash N : B
    \justifies
      \ctxt{\Updelta}{\Gamma} \vdash \letbox{u}{M}{N} : B
  \end{prooftree}
\] along with the reduction $\letbox{u}{\ibox{M}}{N} \rightarrow
N[M/u]$.

Even though the information flow properties of the Davies-Pfenning
calculus are intuitive, it is not at all evident how to express
them as a noninterference theorem: whilst in Moggi's calculus the
monad $T$ blocked all information flow out of it, the
Davies-Pfenning calculus allows \emph{some} flow into the modal
types, as witnessed by the non-trivial morphism $\textsf{Bool}
\rightarrow \Box\textsf{Bool}$ constructed above.

The way out of this impasse is to consider the dual of the
statement used to prove noninterference for the metalanguage,
namely Proposition \ref{prop:cocoint}(2). In the place of
booleans, we will use the \emph{codiscrete booleans}
$\textsf{Bool}_\nabla$.  These are introduced explicitly by \[
    \ctxt{\Updelta}{\Gamma} \vdash 
      \texttt{tt} : \textsf{Bool}_\nabla
  \qquad
    \ctxt{\Updelta}{\Gamma} \vdash 
      \texttt{ff} : \textsf{Bool}_\nabla
\] and the elimination rule \[
  \begin{prooftree}
      \Gamma \vdash M : \textsf{Bool}_\nabla
        \quad
      \Gamma \vdash E_0 : C
        \quad
      \Gamma \vdash E_1 : C
        \quad
      \text{$C$ is codiscrete}
    \justifies
      \Gamma \vdash \textsf{if } M \textsf{ then } E_1 \textsf{
      else } E_0 : C
  \end{prooftree}
\] with the same computation rules as before. The side condition
that $C$ is codiscrete is defined by (I) $\textsf{unit}$ (if we
have it) and $\textsf{Bool}_\nabla$ are codiscrete; (II) if $A$
and $B$ are codiscrete, then so is $A \times B$; (III) if $B$ is
codiscrete, then so is $A \rightarrow B$. We will discover the
roots of this definition in \S\ref{sec:refl}. For now, let us say
that the intended meaning is that, when interpreted in
$\mathsf{CSet}_\mathcal{L}$, the object $\sem{C}{}$ will be
isomorphic to $\blacklozenge \sem{C}{}$. 

Now, let us note that the Davies-Pfenning calculus satisfies
confluence, and strong normalisation \citep{Kavvos2017b,
Kavvos2017c}, and hence canonicity, with canonical forms at type
$\Box A$ of the form $\ibox{M}$, with $M$ is a canonical form at
type $A$. It is straightforward to re-establish canonicity after
the addition of codiscrete booleans. The Davies-Pfenning calculus
also has a standard interpretation in any CCC with a
product-preserving comonad on it: see e.g. \citep{Hofmann1999,
Kavvos2017b}. we may again show adequacy:

\begin{lem}[Davies-Pfenning Adequacy]
  \label{lem:dpadeq}
  Suppose we have a categorical model for the Davies-Pfenning
  calculus, as well as a ground type $G$ that satisfies canonicity
  along with an injective interpretation in that model, so that \[
    \sem{ \vdash \textsf{c}_i : G}{}
      =
    \sem{ \vdash \textsf{c}_j : G}{}
      : \mathbf{1} \rightarrow \sem{G}{}
      \quad\Longrightarrow\quad
    \textsf{c}_i \equiv \textsf{c}_j
  \] Then this interpretation is \emph{adequate for $\Box G$}, in
  the sense that \[
    \sem{\ \vdash M : \Box G}{}
    = \sem{\ \vdash \ibox{\textsf{c}_i} : \Box G}{}
      \quad\Longrightarrow\quad
    \vdash M = \ibox{\textsf{c}_i} : \Box G
  \]
\end{lem}
\begin{proof}
  The canonical forms of type $\Box G$ are precisely those of the
  form $\ibox{\textsf{c}_i}$. We therefore have that $\vdash M =
  \ibox{\textsf{c}_i} : \Box G$ for some $\textsf{c}_i$. But then,
  if $\sem{M}{} = \sem{\ibox{\textsf{c}_j}}{}$ we have---writing
  $(-)^\ast$ for the co-Kleisli extension and using
  soundness---that \[
    \sem{\textsf{c}_j}{\ast}
    = \sem{\ibox{\textsf{c}_j}}{}
    = \sem{M}{}
    = \sem{\ibox{\textsf{c}_i}}{}
    = \sem{\textsf{c}_i}{\ast}
  \] Post-composing with the counit (see \cite[Prop. 4, \S
  7.3.2]{Kavvos2017c}) we obtain $\sem{\textsf{c}_j}{} =
  \sem{\textsf{c}_i}{}$, and hence $\textsf{c}_i \equiv
  \textsf{c}_j$.
\end{proof}

By Corollary \ref{cor:coh}, we know that $\Box :
\mathbf{CSet}_\mathcal{L} \longrightarrow
\mathbf{CSet}_\mathcal{L}$ is a product-preserving comonad, and
thus a model of the Davies-Pfenning calculus. Moreover, we can
interpret $\textsf{Bool}_\nabla$ by $\nabla\mathbb{B}$. This is a
little harder to see: the reason is that the adjunction $U \dashv
\nabla$ is a \emph{reflection}, as $\nabla$ is full and faithful.
Thus, writing $\eta_A : A \rightarrow \nabla U A$ for the unit, we
have a universal property;\footnote{This is derived from the fact
$\eta_A$ is a universal arrow and $\nabla$ is full and faithful.}
it is that any $f : A \rightarrow \nabla X$ can be uniquely
factorised through $\eta_A$: \[
  \begin{tikzcd}
    A
      \arrow[r, "\eta_A"]
      \arrow[dr, "f", swap]
    & \nabla U A
      \arrow[d, "\hat{f}", dotted]
    \\
    & \nabla X
  \end{tikzcd}
\] Specialising this to $A = \Updelta\mathbb{B} \cong \mathbf{1} +
\mathbf{1}$, we have that any $f : \mathbf{1} + \mathbf{1}
\rightarrow \nabla X$ can be uniquely factorised through $\nabla U
\Updelta\mathbb{B} \cong \nabla\mathbb{B}$: \[
  \begin{tikzcd}
    \mathbf{1} + \mathbf{1} 
      \arrow[r, "\eta_A"]
      \arrow[dr, "f", swap]
    & \nabla \mathbb{B}
      \arrow[d, "\hat{f}", dotted]
    \\
    & \nabla X
  \end{tikzcd}
\] This is exactly the elimination rule for
$\textsf{Bool}_\nabla$: to define a function out of it one must
tell what it does on the two constants (the components of
$\mathbf{1} + \mathbf{1}$) by giving a term for each in a
codiscrete type $C$, which is then of the right form $\nabla X$,
as $\sem{C}{} \cong \blacklozenge \sem{C}{} = \nabla (U C)$.

It is straightforward to show that the above interpretation is
sound, and satisfies the assumptions of Lemma \ref{lem:dpadeq}. We
then have

\begin{thm}[Noninterference for Davies-Pfenning]
  Let $\ctxt{\cdot}{x : \textsf{Bool}_\nabla} \vdash M : \Box G$
  for any ground type $G$.  Then, for any $\vdash E, E' :
  \textsf{Bool}_\nabla$, we have \[
    \vdash M[E/x] = M[E'/x] : \Box G
  \]
\end{thm}
\begin{proof}
  As $\mathbb{B}$ is non-empty, we have by Proposition
  \ref{prop:cocoint}(2) that \[
    f \myeq \sem{\ctxt{\cdot}{x : \textsf{Bool}_\nabla} 
      \vdash M : \Box G}{} : 
        \nabla \mathbb{B} \rightarrow \Box \sem{G}{}
  \] is a constant function in $\mathbf{CSet}_\mathcal{L}$, so \[
      \sem{M[E/x]}{} 
    = 
      f  \circ \sem{E}{}
    = 
      f  \circ \sem{E'}{}
    =
      \sem{M[E'/x]}{} 
  \] for any $\vdash E, E' : \textsf{Bool}_\nabla$. By adequacy,
  it follows that $\vdash M[E/x] = M[E'/x] : \Box G$.
\end{proof}

\section{Levelled Cohesion}
  \label{sec:coh2}

Whilst useful for discussing the properties of monadic and
comonadic modalities, the results presented above are not
particularly interesting in terms of information flow. For each
classified set $X$, the set $\Box X$ is a version of it where
everything is visible, whereas $\blacklozenge X$ is a version of
it where nothing is. Might we be able to use the same ideas about
pre-cohesion to make something similar, yet far more expressive?
The key lies in relating pre-cohesion to the set of labels
$\mathcal{L}$, of which we have not yet made any use.

Given a subset $\pi \subseteq \mathcal{L}$ of labels, there is a
`partially forgetful' functor, \[
  U_\pi : \textbf{CSet}_\mathcal{L} \longrightarrow
  \textbf{CSet}_{\mathcal{L} - \pi}
\] which maps classified sets over $\mathcal{L}$ to classified
sets over $\mathcal{L} - \pi$ by forgetting $\rel{\ell}$ for $\ell
\in \pi$. As before, there are two ways to define a set classified
over $\mathcal{L}$ when given a set classified over $\mathcal{L} -
\pi$.

The first one is given by the functor \[
  \Updelta_\pi : \textbf{CSet}_{\mathcal{L} - \pi} \longrightarrow
  \textbf{CSet}_\mathcal{L}
\] $\Updelta_\pi X$ has the same carrier as $X$, and the same
relations $\rel{\ell}$ for $\ell\not\in\pi$. But if $\ell \in
\pi$, then $\rel{\ell}$ of $\Updelta_\pi X$ is the diagonal
relation. We let $\Updelta_\pi$ be the identity on morphisms: all
the relations not in $\pi$ are preserved by $\Updelta_\pi f$, and
the rest are diagonal so they are also trivially preserved. So for
each classified set $X$ over $\mathcal{L} - \pi$, we have a
classified set $\Updelta_\pi X$ over $\mathcal{L}$ which is
\emph{transparent at $\pi$}. It is easy to see that $\Updelta_\pi$
is full and faithful, and that
\begin{prop}
  $\Updelta_\pi \dashv U_\pi$
\end{prop}

The second one is given by the functor \[
  \nabla_\pi : \textbf{CSet}_{\mathcal{L} - \pi} \longrightarrow
  \textbf{CSet}_\mathcal{L}
\] which, this time, adds the complete relation as $\rel{\ell}$
for each $\ell\in\pi$, and is also the identity on morphisms.
Thus, for each classified set $X$ over $\mathcal{L} - \pi$, we
have a classified set $\nabla_\pi X$ over $\mathcal{L}$ which is
\emph{opaque at $\pi$}. It is also easy to see that it
$\nabla_\pi$ is full and faithful, and that \begin{prop} $U_\pi
\dashv \nabla_\pi$
\end{prop}

It remains to consider connected components. Suppose we have a
morphism $f : X \rightarrow \Updelta_\pi Y$. For each $\ell \in
\pi$, the relation $\Updelta_\pi Y$ is reflexivity.  Hence, for
$\ell \in \pi$, \[
  x \rel{\ell} x'\ \text{ (in $X$)}\
    \Longrightarrow
  f(x) \rel{\ell} f(x')\ \text{ (in $\Updelta_\pi Y$)}\
    \Longrightarrow
  f(x) = f(x')
\] Thus, if $x$ and $x'$ are indistinguishable at level $\ell \in
\pi$, then $f$ collapses them to a single element. As before, we
would also like to phrase this in terms of quotients. Let
$\rel{\pi}^\star$ to be the reflexive, symmetric, transitive
closure of $\bigcup_{\ell\in\pi} \rel{\ell}$. The construction of
the quotient set $\bars{X}/R_\pi^\star$ extends to a functor \[
  C_\pi : \textbf{CSet}_\mathcal{L} \longrightarrow
  \textbf{CSet}_{\mathcal{L}-\pi}
\] by letting $C_\pi(X)$ be the classified set with carrier
$\bars{X} / R_\pi^\star$, and, for $\ell \in \mathcal{L} - \pi$,
\[ 
  [b] \rel{\ell} [b']\ \text{ in $C_\pi(X)$ }
    \quad\Longleftrightarrow\quad
  \exists x \in [b].\ \exists y \in [b'].\ x \rel{\ell} y\ \text{
  in $X$ }
\] We let $C_\pi f \myeq f_\pi^\star : \bars{X} / R_\pi^\star
\rightarrow \bars{Y}/R_\pi^\star$, where \[
  f_\pi^\star([x]) \myeq [f(x)]
\] $f_\pi^\star$ is well-defined, as---by the same argument as in
\S\ref{sec:coh}---$f$ preserves $R_\pi^\star$. If $[b] \rel{\ell}
[b']$, there exist $x$ and $y$ such that $x \mathrel{R_\pi^\star}
b$, $y \mathrel{R_\pi^\star} b'$, and $x \rel{\ell} y$. But $f$
preserves all of these relations, so $[f(b)] \rel{\ell} [f(b')]$.
So $f_\pi^\star$ is actually a morphism $C_\pi(X) \rightarrow
C_\pi(Y)$. A similar argument to the one in \S\ref{sec:coh} takes
$f : X \rightarrow \Updelta_\pi Y$ to a unique $\hat{f} : C_\pi X
\rightarrow Y$, and hence

% Like before, define the relation \[
%   x \sim_f x'
%     \quad\myeq\quad
%   f(x) = f(x')
% \] This is an equivalence relation, and $\bigcup_{\ell \in \pi}
% \rel{\ell}\ \subseteq\ \mathrel{\sim_f}$. But $R_\pi^\star$ is
% the least equivalence relation such that this inclusion holds,
% above is true, so $R_\pi^\star \subseteq\ \mathrel{\sim_f}$.
% Hence, $Uf : \bars{X} \rightarrow Y$, which we showed respects
% $\bigcup_{\ell \in \pi} \rel{\ell}$, can be uniquely factored as
% \[
%   \begin{tikzcd}
%     \bars{X}
%       \arrow[r, "\eta_X"]
%       \arrow[dr, "Uf", swap]
%     & \bars{X} / R_\pi^\star 
%       \arrow[d, "\hat{f}", dotted] \\
%     & Y
%   \end{tikzcd}
% \] where $\eta_X$ is the quotient map, and $\hat{f}([x]) \myeq
% f(x)$. So we can uniquely take $f : X \rightarrow \Updelta_\pi Y$
% to $\hat{f} : C_\pi X \rightarrow Y$, and

\begin{prop}
  $C_\pi : \mathbf{CSet}_\mathcal{L} \rightarrow
  \mathbf{CSet}_{\mathcal{L}-\pi}$ is left adjoint to
  $\Updelta_\pi : \mathbf{CSet}_{\mathcal{L}-\pi} \rightarrow
  \mathbf{CSet}_\mathcal{L}$.
\end{prop}

\noindent It is easy to see that this preserves products, and that
the nullstellensatz holds as before (the natural isomorphisms
showing $\Updelta_\pi \dashv U_\pi \dashv \nabla_\pi$ are
identities on the hom-sets). In total:

\begin{thm}
  \label{thm:cohlev}
  $\mathbf{CSet}_\mathcal{L}$ is pre-cohesive relative to
  $\mathbf{CSet}_{\mathcal{L}-\pi}$.
\end{thm}

It is worth reiterating some of the things that we automatically
learn from this fact:
\begin{cor} For a given $\pi \subseteq \mathcal{L}$,
  \begin{enumerate}
    \item
      $U_\pi : \mathbf{CSet}_\mathcal{L} \longrightarrow
      \mathbf{CSet}_{\mathcal{L} - \pi}$ preserves (co)limits.
    \item
      $\Updelta_\pi : \mathbf{CSet}_{\mathcal{L} - \pi}
      \longrightarrow \mathbf{CSet}_\mathcal{L}$ preserves
      (co)limits.
    \item
      $\nabla_\pi : \mathbf{CSet}_{\mathcal{L} - \pi}
      \longrightarrow \mathbf{CSet}_\mathcal{L}$ preserves limits.
    \item
      $\Box_\pi \myeq \Updelta_\pi U_\pi :
      \textbf{CSet}_\mathcal{L}\longrightarrow
      \textbf{CSet}_\mathcal{L}$ is an idempotent comonad that
      preserves finite (co)limits.
    \item
      $\blacklozenge_\pi \myeq \nabla_\pi U_\pi :
      \textbf{CSet}_\mathcal{L} \longrightarrow
      \textbf{CSet}_\mathcal{L}$ is an idempotent monad that
      preserves finite limits.
    \item
      $\int_\pi \myeq \Updelta_\pi C_\pi :
      \textbf{CSet}_\mathcal{L}\longrightarrow
      \textbf{CSet}_\mathcal{L}$ is an idempotent monad.
      preserves products and colimits.
    \item
      $\int_\pi \dashv \Box_\pi \dashv \blacklozenge_\pi$
  \end{enumerate}
\end{cor}

\noindent As before, we can use this structure to prove abstract
noninterference theorems for $\mathbf{CSet}_\mathcal{L}$, which
will now be more expressive than before. Before we take on that
task, however, we want to discuss two aspects of
$\mathbf{CSet}_\mathcal{L}$. First, we want to identify the
reflective and (co)reflective subcategories identified by
$\Box_\pi$ and $\blacklozenge_\pi$. And, second, we want to prove
once and or all that $\blacklozenge_\pi$ is a strong monad.

\subsection{(Co)reflective subcategories}
  \label{sec:(co)refl}

In cohesive settings we have a string of adjoints $\Updelta \dashv
U \dashv \nabla$, where $\Updelta$ and $\nabla$ are full and
faithful. This is to say that both $\Updelta$ and $\nabla$ are a
kind of `inclusion,' and that---up to equivalence---they exhibit
certain subcategories of $\mathcal{E}$. In particular, $\nabla$
exhibits a \emph{reflective subcategory} of objects $Y$ such that
$\blacklozenge Y \cong Y$, and $\Updelta$ exhibits a
\emph{coreflective subcategory} of $\mathcal{E}$, viz. the full
subcategory consisting of objects $X$ such that $\Box X \cong X$.
Moreover, such adjunctions generate \emph{idempotent (co)monads},
in that their (co)multiplications $\Box \Rightarrow \Box^2$ and
$\blacklozenge^2 \Rightarrow \blacklozenge$ are isomorphisms
\citep[Prop. 4.3.2]{Borceux1994}.

It is illuminating to look at each of these cases for the
pre-cohesion of $\textbf{CSet}_\mathcal{L}$ relative to
$\textbf{CSet}_{\mathcal{L} - \pi}$.

\paragraph{Reflection}
  \label{sec:refl}

If $X$ is a classified set such that $X \cong \blacklozenge_\pi
X$, then $\rel{\ell}$ is the complete relation at all $\ell \in
\pi$. This leads us to the following definition:

\begin{defn}[Protection]
  Let $\pi \subseteq \mathcal{L}$ be a set of labels. A classified
  set $X$ is \emph{protected at $\pi$} just if $R_\ell = \bars{X}
  \times \bars{X}$ for all $\ell \in \pi$.
\end{defn}

\noindent Thus the monad $\blacklozenge_\pi$ induces a reflective
subcategory $\textbf{CSet}_{\mathcal{L}, \text{p}, \pi}$ of
\emph{classified sets protected at $\pi$}. In other words, there
is a functor $\mathcal{R}_\pi : \textbf{CSet}_\mathcal{L}
\longrightarrow \textbf{CSet}_{\mathcal{L}, \text{p}, \pi}$ that
is left adjoint to the inclusion: \[
  \begin{tikzcd}
    \textbf{CSet}_{\mathcal{L}, \text{p}, \pi}
      \arrow[r, bend right, "i"{below}, ""{above, name=B}, hook]
    & \textbf{CSet}_\mathcal{L}
      \arrow[l, bend right, "\mathcal{R}_\pi"{above}, ""{below, name=A}]
        \arrow[from=A, to=B, symbol=\dashv]
  \end{tikzcd}
\] $\mathcal{R}_\pi$ replaces $\rel{\ell}$ with the complete
relation for every $\ell \in \pi$. In addition, we have that
$\blacklozenge_\pi = i \circ \mathcal{R}_\pi$. Finally, it is easy
to see that $\blacklozenge_\pi$ is a \emph{strictly idempotent}
monad: the multiplication is the identity, as $X =
\blacklozenge_\pi X$ whenever $X$ is protected at $\pi$.

Interestingly,
\begin{prop}
  $\mathcal{R}_\pi : \textbf{CSet}_\mathcal{L} \longrightarrow
  \textbf{CSet}_{\mathcal{L}, \text{p}, \pi}$ preserves finite
  products.
\end{prop} 

\noindent This is because $\textbf{CSet}_{\mathcal{L}, \text{p},
\pi} \cong \textbf{CSet}_{\mathcal{L} - \pi}$, and up to that
equivalence $\mathcal{R}_\pi$ is just $U_\pi$, and $i$ is just
$\nabla_\pi$. By a theorem of category theory---see e.g. \citep[\S
A4.3.1]{Johnstone2003}---if a reflector preserves finite products,
as $\mathcal{R}_\pi$ does, then the reflective subcategory is an
\emph{exponential ideal}. That is,

\begin{cor}
  If $A$ is a classified set, and $B$ is protected at $\pi$, then
  $B^A$ is protected at $\pi$.
\end{cor}

\noindent This is what underlies the definition of
\emph{codiscrete types} we gave in \S\ref{sec:dp},
and which was introduced as a unmotivated definition in
\cite{Abadi1999}.

\paragraph{Coreflection}

Dually, the second subcategory consists of classified sets $X$
over $\mathcal{L}$ such that $X \cong \Box_\pi X$. Each
$\rel{\ell}$ in $X$ for $\ell\in\pi$ is the diagonal relation. The
corresponding definition is:

\begin{defn}[Visibility]
  Let $\pi \subseteq \mathcal{L}$ be a set of labels. A classified
  set $X$ is \emph{visible at $\pi$} just if $R_\ell = \setcomp{
  (x, x) }{ x \in \bars{X} }$ for every $\ell \in \pi$.
\end{defn}

This full coreflective subcategory induced $\Box_\pi$ is notated
$\textbf{CSet}_{\mathcal{L}, \text{v}, \pi}$, and its objects are
those \emph{classified sets that are visible at $\pi$}. There is a
functor \[
  \mathcal{D}_\pi : \textbf{CSet}_\mathcal{L} \longrightarrow
  \textbf{CSet}_{\mathcal{L}, \text{v}, \pi}
\] which returns the classified set $\mathcal{D}_\pi X$ with the
same carrier, but if $\ell \in \pi$ then $R_\ell$ is the diagonal
relation. It is also the identity on morphisms. This functor is
right adjoint to the inclusion: \[
  \begin{tikzcd}
    \textbf{CSet}_\mathcal{L}
      \arrow[r, bend right, "D_\pi"{below}, ""{above, name=B}]
    & \textbf{CSet}_{\mathcal{L}, \text{v}, \pi}
      \arrow[l, bend right, "i"{above}, ""{below, name=A}, hook]
        \arrow[from=A, to=B, symbol=\dashv]
  \end{tikzcd}
\] and of course $\Box_\pi = i \circ \mathcal{D}_\pi$. In a manner
similar to previous one, $\Box_\pi$ is also a \emph{strictly
idempotent comonad}, as $X = \Box_\pi X$ whenever $X$ is visible
at $\pi$.

\subsection{Redaction is a strong monad}
  \label{sec:strong}

We record here a fact that we have already used and will use
again, namely that

\begin{prop} \label{prop:strong}
  $\blacklozenge_\pi$ is a \emph{strong monad}: that is, there
  exists a natural transformation \[
    t_{A, B} : A \times \blacklozenge_\pi B \rightarrow
    \blacklozenge_\pi (A \times B)
  \] such that the following diagrams commute: \[
    \begin{tikzcd}
      \mathbf{1} \times \blacklozenge_\pi A
        \arrow[r, "{t_{\mathbf{1}, A}}"]
        \arrow[d, "\pi_2", swap]
      & \blacklozenge_\pi(\mathbf{1} \times A)
        \arrow[dl, "\blacklozenge_\pi \pi_2"] \\
      \blacklozenge_\pi A
      & 
    \end{tikzcd}
    \begin{tikzcd}
      (A \times B) \times \blacklozenge_\pi C
        \arrow[r, "{t_{A \times B, C}}"]
        \arrow[d, "\cong", swap]
      & \blacklozenge_\pi\left((A \times B) \times C\right)
        \arrow[dr, "\cong"]
      & \\
      A \times (B \times \blacklozenge_\pi C)
        \arrow[r, "id_A \times t_{B, C}", swap]
      & A \times \blacklozenge_\pi(B \times C)
        \arrow[r, "{t_{A, B \times C}}", swap]
      & \blacklozenge_\pi\left(A \times (B \times C)\right)
    \end{tikzcd}
  \] \[
    \begin{tikzcd}
      A \times B
        \arrow[dr, "\eta_{A \times B}"]
        \arrow[d, "id_A \times \eta_B", swap]
      & & \\
      A \times \blacklozenge_\pi B
        \arrow[r, "{t_{A, B}}"]
      & \blacklozenge_\pi(A \times B)
      & \\
      A \times \blacklozenge^2_\pi B
        \arrow[u, "id_A \times \mu_B"]
        \arrow[r, "{t_{A, \blacklozenge_\pi B}}", swap]
      & \blacklozenge_\pi\left(A \times \blacklozenge_\pi B\right)
        \arrow[r, "{\blacklozenge_\pi t_{A, B}}", swap]
      & \blacklozenge^2_\pi (A \times B)
        \arrow[ul, "\mu_{A \times B}", swap]
    \end{tikzcd}
  \]
\end{prop}
\begin{proof}
  The components $t_{A, B}$ are identity functions on $\bars{A
  \times \blacklozenge_\pi B} = \bars{A} \times \bars{B} =
  \bars{\blacklozenge_\pi (A \times B)}$. These preserve all the
  relations at $\ell \in \pi$, as $\blacklozenge_\pi(A \times B)$
  is protected at $\pi$. If $\ell \not\in\pi$, then $(a, b)
  \rel{\ell} (a', b')$ in $A \times \blacklozenge_\pi B$ means
  that $a \rel{\ell} a'$ in $A$ and $b \rel{\ell} b'$ in $B$, so
  $(a, b) \rel{\ell} (a', b')$ in $A \times B$ and hence in
  $\blacklozenge_\pi (A \times B)$.  The diagrams commute as all
  the arrows excluding projections and associativities (but
  including components of $t$, $\eta$ and $\mu$ and products
  thereof) are identities on carrier sets.
\end{proof}

\subsection{Stacking pre-cohesions and inter-level reasoning}
  \label{sec:stack}

In fact, it is not hard to generalise the results in the above
section to the fact that we can `stack' such pre-cohesions on top
of one another. For example, by applying Theorem \ref{thm:cohlev}
twice, if $\pi' \subseteq \pi \subseteq \mathcal{L}$, we can
obtain two pre-cohesions: \[
  \begin{tikzcd}
      \textbf{CSet}_\mathcal{L}
        \arrow[d, shift right=50pt, 
          "C_{\pi \subseteq \mathcal{L}}", swap]
        \arrow[d, shift left=20pt, 
          "U_{\pi \subseteq \mathcal{L}}", "\dashv"{name=A, left}]
      \\
      \textbf{CSet}_{\pi}
        \arrow[u, shift left=20pt, 
          "\Updelta_{\pi \subseteq \mathcal{L}}", swap, "\dashv"{name=B,left}]
        \arrow[u, shift right=60pt, 
          "\nabla_{\pi \subseteq \mathcal{L}}", swap, "\dashv"{name=B,left}]
        \arrow[d, shift right=50pt, 
          "C_{\pi' \subseteq \pi}", swap]
        \arrow[d, shift left=20pt, 
          "U_{\pi' \subseteq \pi}", "\dashv"{name=A, left}]
      \\
      \textbf{CSet}_{\pi'}
        \arrow[u, shift left=20pt,
          "\Updelta_{\pi' \subseteq \pi}", swap,
        "\dashv"{name=B,left}]
        \arrow[u, shift right=60pt, 
          "\nabla_{\pi' \subseteq \pi}", swap, "\dashv"{name=B,left}]
    \end{tikzcd}
\] which compose to what we previously denoted $C_{\mathcal{L} -
\pi'} \dashv \Updelta_{\mathcal{L} - \pi'} \dashv U_{\mathcal{L} -
\pi'} \dashv \nabla_{\mathcal{L} - \pi'}$. For example,
\begin{align*}
  U_{\mathcal{L} - \pi'} &= 
          U_{\pi' \subseteq \pi} \circ U_{\pi' \subseteq \mathcal{L}}
    : \mathbf{CSet}_\mathcal{L} \longrightarrow \mathbf{CSet}_{\pi'} \\
  \nabla_{\mathcal{L} - \pi'} &=
          \nabla_{\pi \subseteq \mathcal{L}}
    \circ \nabla_{\pi' \subseteq \pi}
    : \mathbf{CSet}_{\pi'} \longrightarrow \mathbf{CSet}_\mathcal{L}
\end{align*} and so forth. This forms a functor \[
  \mathcal{P}(\mathcal{L})^{\text{op}}
    \longrightarrow
  \mathbf{Precoh}
\] from the opposite powerset lattice of $\mathcal{L}$ to
$\mathbf{Precoh}$, whose morphisms are strings of adjoints $(C,
\Updelta, U, \nabla) : \mathcal{E} \rightarrow \mathcal{S}$ that
exhibit $\mathcal{E}$ to be pre-cohesive over $\mathcal{S}$.
The functor maps the unique arrow $\alpha : \pi \subseteq \pi'$ to
a pre-cohesion $(C_\alpha, \Updelta_\alpha, U_\alpha,
\nabla_\alpha) : \mathbf{CSet}_{\pi'} \rightarrow
\mathbf{CSet}_{\pi}$. Thus, each $\alpha : \pi \subseteq \pi'$
induces modalities $\Box_\alpha, \blacklozenge_\alpha :
\mathbf{CSet}_{\pi'} \longrightarrow \mathbf{CSet}_{\pi}$, and we
previously wrote $\Box_\pi$ for $\Box_{\alpha : \mathcal{L}
- \pi \subseteq \mathcal{L}}$.

This functorial structure satisfies a number of
strange-looking---yet very intuitive in terms of information
flow---properties. To begin, we recall the isomorphism $U_\alpha\
\Delta_\alpha \cong U_\alpha\ \nabla_\alpha \cong \textsf{Id}$,
which was a consequence of pre-cohesion (Proposition
\ref{cor:coh}), and holds \emph{on-the-nose} in classified sets.
This has the following consequences regarding the induced
modalities.

\begin{prop} 
  \label{prop:absorption}
    If $\gamma : \pi'' \subseteq \pi'$ and $\alpha : \pi'
    \subseteq \pi$, then:
  \begin{enumerate}
    \item
        $\blacklozenge_{\alpha \circ \gamma}\ \Box_\alpha =
        \blacklozenge_{\alpha \circ \gamma} : \mathbf{CSet}_\pi
        \longrightarrow \mathbf{CSet}_\pi$
    \item
        $\Box_{\alpha \circ \gamma}\ \blacklozenge_\alpha =
        \Box_{\alpha \circ \gamma} : \mathbf{CSet}_\pi
        \longrightarrow \mathbf{CSet}_\pi$
  \end{enumerate}
\end{prop}
\begin{proof} \hfill
  \begin{enumerate}
    \item \begin{align*}
        \blacklozenge_{\alpha \circ \gamma}\ \circ \Box_\alpha
          &=
              \nabla_{\alpha \circ \gamma}\ U_{\alpha \circ \gamma}\ 
              \Updelta_{\alpha}\ U_{\alpha}
            & \text{ by definition }\\
          &=
              \nabla_{\alpha \circ \gamma}\ U_\gamma\ U_\alpha\
              \Updelta_{\alpha}\ U_{\alpha}
            & \text{ by functoriality and contravariance of $U$ }\\
          &=
              \nabla_{\alpha \circ \gamma}\ U_\gamma\ U_{\alpha}
            & \text{ by the fundamental corollary}\\
          &=
              \nabla_{\alpha \circ \gamma}\ U_{\alpha \circ \gamma}\
            & \text{ by functoriality and contravariance of $U$ }\\
          &=
              \blacklozenge_{\alpha \circ \gamma}
            & \text{ by definition}\\
      \end{align*}
    \item Similar to (1).
  \end{enumerate}
\end{proof}

\noindent Intuitively, the first of these results says that
declassifying things at some levels, and then redacting more than
what was declassified is exactly the same as redacting all at
once. Note that we used no particular properties about the model
of classified sets, apart from the strictness of $U\Updelta \cong
U\nabla \cong \textsf{Id}$.

Next, the forgetful functor $U$ and the discretisation functor
$\Updelta$ sometimes commute. More specifically, if we have a
pullback diagram in $\mathcal{P}(\mathcal{L})$ 
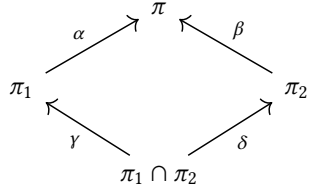
\[
  \begin{tikzcd}
     & \pi & \\
    \pi_1
      \arrow[ur, "\alpha"]
    &  &
    \pi_2
      \arrow[ul, "\beta" swap] \\
    & 
    \pi_1 \cap \pi_2
      \arrow[ul, "\gamma"] 
      \arrow[ur, "\delta", swap]
  \end{tikzcd}
\] then it is easy to check that the following functor diagram
commutes on-the-nose: 
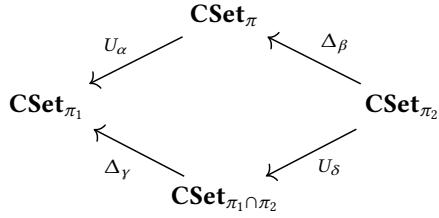
\[
  \begin{tikzcd}
     & \textbf{CSet}_\pi
      \arrow[dl, "U_{\alpha}", swap]
      & \\
    \textbf{CSet}_{\pi_1}
    &  &
    \textbf{CSet}_{\pi_2}
      \arrow[ul, "\Updelta_{\beta}", swap]
      \arrow[dl, "U_{\delta}"]
      \\
    & 
    \textbf{CSet}_{\pi_1 \cap \pi_2}
      \arrow[ul, "\Updelta_{\gamma}"]
  \end{tikzcd}
\] That is: if, starting from $\pi_2$, we equip the extra labels
with discrete cohesion, and then forget everything down to
$\pi_1$, we have not changed any of the labels of $\pi_1 \cap
\pi_2$. Thus, we might first forget, and then discretise up to
$\pi_1$. The same holds of codiscretisation, so we obtain \[
  U_\alpha \nabla_\beta = \nabla_\gamma U_\delta
\] A sufficient condition for the above is that the components of
the pasting diagram 
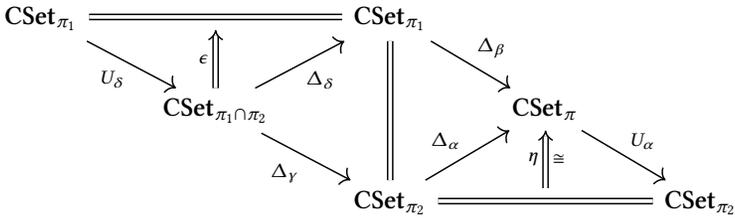
\[
  \begin{tikzcd}
    \mathbf{CSet}_{\pi_1}
      \arrow[dr, "U_\delta", swap]
      \arrow[rr, equal, ""{name=U, below}]
  & 
  & \mathbf{CSet}_{\pi_1}
    \arrow[rd, "\Updelta_\beta"]
    \arrow[dd, equal]
  & 
  & \\
  & \mathbf{CSet}_{\pi_1 \cap \pi_2}
    \arrow["\epsilon", Rightarrow, to=U]
    \arrow[ur, "\Updelta_\delta", swap]
    \arrow[dr, "\Updelta_\gamma", swap]
  &
  & \mathbf{CSet}_\pi
    \arrow[rd, "U_\alpha"]
  & \\
  &
  & \mathbf{CSet}_{\pi_2}
    \arrow[ur, "\Updelta_\alpha"]
    \arrow[rr, equal, ""{name=W, above}]
    \arrow[ur, "\eta", Rightarrow, from=W, "\cong"{swap}]
  & 
  & \mathbf{CSet}_{\pi_2}
  \end{tikzcd}
\] are identities---and they indeed are in our case: the
components of $\eta$ and $\epsilon$ are identity functions, so
this reduces to simply checking $\Updelta_\gamma\ U_\delta =
U_\alpha\ \Delta_\beta$ again.

Secondly, if we discretise and codiscretise in disjoint
regions, then these operations can be swapped. Namely, given a
pullback diagram like above, we also have that 
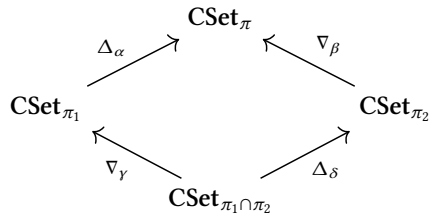
\[
  \begin{tikzcd}
     & \mathbf{CSet}_{\pi} & \\
    \mathbf{CSet}_{\pi_1}
      \arrow[ur, "\Updelta_\alpha"]
    &  &
    \mathbf{CSet}_{\pi_2}
      \arrow[ul, "\nabla_\beta" swap] \\
    & 
    \mathbf{CSet}_{\pi_1 \cap \pi_2}
      \arrow[ul, "\nabla_\gamma"] 
      \arrow[ur, "\Updelta_\delta", swap]
  \end{tikzcd}
\] commutes. Again, it suffices that the following pasting
diagram composes to an identity: 
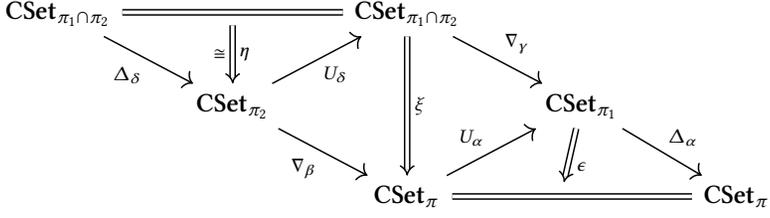
\[
  \begin{tikzcd}
    \mathbf{CSet}_{\pi_1 \cap \pi_2}
      \arrow[dr, "\Updelta_\delta", swap]
      \arrow[rr, equal, ""{name=U, below}]
  & 
  & \mathbf{CSet}_{\pi_1 \cap \pi_2}
    \arrow[rd, "\nabla_\gamma"]
    \arrow[dd, "\xi", Rightarrow]
  & 
  & \\
  & \mathbf{CSet}_{\pi_2}
    \arrow["\eta", "\cong"{swap}, Rightarrow, from=U]
    \arrow[ur, "U_\delta", swap]
    \arrow[dr, "\nabla_\beta", swap]
  &
  & \mathbf{CSet}_{\pi_1}
    \arrow[rd, "\Updelta_\alpha"]
    \arrow[dd, "\epsilon", Rightarrow, to=W]
  & \\
  &
  & \mathbf{CSet}_\pi
    \arrow[ur, "U_\alpha"]
    \arrow[rr, equal, ""{name=W, above}]
  & 
  & \mathbf{CSet}_\pi
  \end{tikzcd}
\] where $\xi$ is a vertion of the preceding pasting diagram for
$\nabla$.

These `basic laws' allow us to prove many more that relate the
different cohesive structures.

\begin{prop}
  \label{prop:multimodal}
  Let $\begin{tikzcd}
     & \pi & \\
    \pi_1
      \arrow[ur, "\alpha"]
    &  &
    \pi_2
      \arrow[ul, "\beta" swap] \\
    & 
    \pi_1 \cap \pi_2
      \arrow[ul, "\gamma"] 
      \arrow[ur, "\delta", swap]
  \end{tikzcd}$ be a pullback diagram. Then
  \begin{enumerate}
    \item
      $\Box_\gamma\ U_\alpha = U_\alpha\ \Box_\beta$
    \item
      $\blacklozenge_\gamma\ U_\alpha = U_\alpha\ \blacklozenge_\beta$
    \item
      $\Box_{\alpha \circ \gamma} = \Box_\alpha\  \Box_\beta$
    \item
      $\blacklozenge_{\alpha \circ \gamma} = \blacklozenge_\alpha\
      \blacklozenge_\beta$
    \item
      $\Box_\alpha\ \Box_\beta = \Box_\beta\ \Box_\alpha$
    \item
      $\blacklozenge_\alpha\ \blacklozenge_\beta =
      \blacklozenge_\beta\ \blacklozenge_\alpha$
  \end{enumerate}
\end{prop}
\begin{proof} \hfill
  \begin{enumerate}
    \item
      \begin{align*}
        \Box_\gamma\ U_\alpha
          &=
              \Updelta_\gamma\ U_\gamma\ U_\alpha\
            & \text{ by definition }\\
          &=
              \Updelta_\gamma\ U_\delta\ U_\beta\
            & \text{ by functoriality of $U$ }\\
          &=
              U_\alpha\ \Updelta_\beta U_\beta
            & \text{ by the first basic law above }\\
          &= 
              U_\alpha\ \Box_\beta
            & \text { by definition }
      \end{align*}
    \item
      Similar to (1), but with $\nabla$.
    \item
      \begin{align*}
        \Box_{\alpha \circ \gamma}
          &=
              \Updelta_\alpha\ \Updelta_\gamma\ U_\gamma\ U_\alpha
            & \text{ by definition ($U$ contravariant) }\\
          &=
              \Updelta_\alpha\ \Box_\gamma\ U_\alpha
            & \text{ by definition }\\
          &=
              \Updelta_\alpha\ U_\alpha\ \Box_\beta
            & \text{ by (1) }\\
          &=
              \Box_\alpha\ \Box_\beta
            & \text{ by definition }
      \end{align*}
    \item
      Similar to (3), but with $\nabla$.
    \item
      By applying (3) twice: $
        \Box_\alpha\ \Box_\beta
          = \Box_{\alpha \circ \gamma}
          = \Box_{\beta \circ \delta}
          = \Box_\beta\ \Box_\alpha
      $ 
    \item
      Similar to (5).
  \end{enumerate}
\end{proof}

\noindent All this structure allows us to show things about the
original modalities $\Box_\pi, \blacklozenge_\pi :
\mathbf{CSet}_\mathcal{L} \longrightarrow
\mathbf{CSet}_\mathcal{L}$.

\begin{prop} \hfill
  \label{prop:switch}
  \begin{enumerate}
    \item
      If $\pi \cap \pi' = \emptyset$, then $\Box_{\pi} \Box_{\pi'}
      = \Box_{\pi \cup \pi'}$.
    \item
      If $\pi \cap \pi' = \emptyset$, then $\blacklozenge_{\pi}
      \blacklozenge_{\pi'} = \blacklozenge_{\pi \cup \pi'}$.
    \item
      $\Box_{\pi} \Box_{\pi'} = \Box_{\pi \cup \pi'}$
    \item
      $\blacklozenge_{\pi} \blacklozenge_{\pi'} =
      \blacklozenge_{\pi \cup \pi'}$
    \item
      If $\pi \subseteq \pi'$, then $\Box_{\pi'}\
      \blacklozenge_\pi = \Box_{\pi'}$.
    \item
      If $\pi \subseteq \pi'$, then $\blacklozenge_{\pi'}\
      \Box_\pi = \blacklozenge_{\pi'}$.
    \item
      If $\pi \cap \pi' = \emptyset$, then $\Box_{\pi}\
      \blacklozenge_{\pi'} = \blacklozenge_{\pi'}\ \Box_{\pi}$.
    \item
      $\Box_\pi\ \blacklozenge_{\pi'} = \blacklozenge_{\pi' -
      \pi}\ \Box_\pi$.
    \item
      $\blacklozenge_\pi\ \Box_{\pi'} = \Box_{\pi' - \pi}\
      \blacklozenge_\pi$.
  \end{enumerate}
\end{prop}
\begin{proof}
  Notice that $\pi \cap \pi' = \emptyset$ implies that
  $\tiny \begin{tikzcd}
     & \mathcal{L} & \\
    \mathcal{L} - \pi 
      \arrow[ur, "\alpha"]
    &  &
    \mathcal{L} - \pi'
      \arrow[ul, "\beta" swap] \\
    & 
    \mathcal{L} - (\pi \cup \pi')
      \arrow[ul, "\gamma"] 
      \arrow[ur, "\delta", swap]
  \end{tikzcd}$ is a pullback diagram. We can hence use
  Prop. \ref{prop:multimodal}: for (1), we have \[
    \Box_{\pi \cup \pi'}
    = \Box_{\alpha \circ \gamma}
    = \Box_\alpha\ \Box_\beta
    = \Box_\pi\ \Box_{\pi'}
  \] by definition and Prop. \ref{prop:multimodal}(3), and very
  similarly for (2). (3) follows by writing $\pi = \pi_1 \uplus
  (\pi \cap \pi')$ and $\pi' = \pi_2 \uplus (\pi \cap \pi')$ as
  disjoint unions, and then using (1) and strict idempotence to
  compute \[
    \Box_\pi\ \Box_{\pi'}
    = \Box_{\pi_1}\ \Box_{\pi \cap \pi'}\ \Box_{\pi \cap \pi'}\ \Box_{\pi_2}
    = \Box_{\pi_1}\ \Box_{\pi \cap \pi'}\ \Box_{\pi_2}
    = \Box_{\pi_1 \cup (\pi \cap \pi') \cup \pi_2}
  \] which is by definition equal to $\Box_{\pi \cup \pi'}$.
  Again, a similar story for (4).

  (5) and (6) follow from Prop. \ref{prop:absorption}
  (with $\gamma : \mathcal{L} - \pi' \subseteq \mathcal{L} - \pi$
  and $\alpha : \mathcal{L} - \pi \subseteq \mathcal{L}$).

  For (7), we calculate:
    \begin{align*}
      \Box_\pi\ \blacklozenge_{\pi'}
        &=
            \Box_\alpha\ \blacklozenge_{\beta}
          & \text{ by definition }\\
        &=
            \Updelta_\alpha\ U_\alpha\ \nabla_\beta\ U_\beta
          & \text{ by definition }\\
        &=
            \Updelta_\alpha\ \nabla_\gamma\ U_\delta\ U_\beta
          & \text{ by the first basic law }\\
        &=
            \nabla_\beta\ \Updelta_\delta\ U_\delta\ U_\beta
          & \text{ by the second basic law }\\
        &=
            \nabla_\beta\ \Box_\delta\ U_\alpha
          & \text{ by definition }\\
        &=
            \nabla_\beta\ U_\beta\ \Box_\alpha
          & \text{ by Prop. \ref{prop:multimodal}(1) }\\
        &=
            \blacklozenge_\beta\ \Box_\alpha
          & \text{ by definition }
    \end{align*}

  Finally, (8) follows easily by writing $\pi = (\pi \cap \pi')
  \uplus (\pi - \pi')$ and using (4), (7) and (5); and similarly
  for (9).
\end{proof}

All of the above equational laws express very intuitive
properties. For example, (3) says that if we declassify everything
at security levels $\pi \cup \pi'$, we could have done that in two
steps, with either $\pi$ or $\pi'$ first. (5) can be understood to
mean that protecting everything at levels $\pi$ and then
declassifying everything at a larger set of levels $\pi'$ is
exactly the same as declassifying $\pi'$ in one go. (7) allows us
to switch redaction and declassification if they act on disjoint
sets of labels. Finally, (8) and (9) show how to switch them even
when there is overlap.

We have therefore developed an armoury of results about information
flow. But notice that we have used no relational reasoning at all!
We have only relied on a functor $\mathcal{C}^\text{op}
\longrightarrow \mathbf{Precoh}$, and the three equations
\begin{align}
      U_\alpha \Updelta_\beta = \Updelta_\gamma\ U_\delta
        \label{eqn:1} \\
      U_\alpha \nabla_\beta = \nabla_\gamma\ U_\delta
        \label{eqn:2} \\
      \Updelta_\alpha\ \nabla_\gamma = \nabla_\beta\ \Updelta_\delta
        \label{eqn:3}
\end{align} for each pullback diagram $\begin{tikzcd}
  \cdot
    \arrow[r, "\gamma"]
    \arrow[d, "\delta", swap]
  & \cdot
    \arrow[d, "\alpha"]
  \\
  \cdot
    \arrow[r, "\beta", swap]
  & \cdot
\end{tikzcd}$ in $\mathcal{C}$. It is conceivable that these
equations could have a more general standing, especially if we
replace equality with natural isomorphism. All the results in this
section would then hold, but only up to natural isomorphism.

\section{Noninterference II: Multi-modal information flow}
  \label{sec:noninterference2}

The theory developed in the previous section enables us to model
\emph{multi-modal} information flow calculi, which feature
modalities that are \emph{indexed} in some way. The usual way to
do so is to index them over a poset $(\mathcal{L}, \sqsubseteq)$
of security levels, and which is often (but not always) a lattice.
We will focus on two main examples: the \emph{dependency core
calculus} (DCC) of \citet{Abadi1999}, and the \emph{sealing
calculus} of \citet{Shikuma2008}.

We have now moved on to levelled cohesion over an arbitrary set
$\mathcal{L}$ of labels. The first thing we want to note is that
codiscrete contractibility is still satisfied, as long as $\pi
\neq \emptyset$: if $X$ is non-empty, we redact it at some levels
$\pi$, and then take the `view from $\pi$,' we end up with almost
nothing, namely a single connected component, i.e.
$C_\pi(\nabla_\pi X) \cong \mathbf{1}$ For that reason, a levelled
version of Prop. \ref{prop:cocoint} from
\S\ref{sec:noninterference1} holds:

\begin{prop} \label{prop:cocoint2}
  If $\pi \neq \emptyset$ and $A$ is non-empty, then morphisms
  $\blacklozenge_\pi A \rightarrow \Updelta_\pi B$ naturally
  correspond to points $\mathbf{1} \rightarrow B$, and are hence
  constant functions in $\mathbf{CSet}_\mathcal{L}$.
\end{prop}

The proof is the same as before.

\subsection{Dependency Core Calculus}

The \emph{dependency core calculus} of \citet{Abadi1999} is at its
core a version of Moggi's computational metalanguage with multiple
monads $T_\ell$ indexed over $\ell \in \mathcal{L}$, where
$(\mathcal{L}, \sqsubseteq)$ is a \emph{lattice}, also called an
\emph{information flow lattice}.\footnote{It is worth noting that,
curiously, the meet and join operations are never used in the
study of the DCC.} The introduction rule is exactly that of Moggi:
\[
  \begin{prooftree}
      \Gamma \vdash M : A
    \justifies
      \Gamma \vdash [M]_\ell : T_\ell A
  \end{prooftree}
\] The elimination rule is modified slightly: \[
  \begin{prooftree}
      \Gamma \vdash M : T_\ell A
    \quad
      \Gamma , x : A \vdash N : B
    \quad
      $B$ \text{ is protected at $\ell$ }
    \justifies
      \Gamma \vdash \textsf{let } x = M \textsf{ in } N : B
  \end{prooftree}
\] Following \citet{Abadi1999}, we say that the type $B$ is
protected at $\ell$ whenever: (I) if $\ell' \sqsubseteq \ell$ then
$T_{\ell}(A)$ is protected at $\ell'$; (II) if $A$ is protected at
$\ell$, then so is $T_{\ell'}(A)$ for any $\ell'$; and (III) if
$A, B$ are protected at $\ell$, then so are $A \times B$ and $C
\rightarrow A$ for any type $C$.

To interpret the DCC\footnote{More specifically: a version of the
DCC without fixpoints, which the original included.} all we need
is \emph{strong, strictly idempotent monad} $T_\ell$ on a CCC
$\mathcal{C}$ for each $\ell \in \mathcal{L}$, such that if $B$ is
a protected type, then $\sem{B}{} = T_\ell \sem{B}{}$ strictly,
i.e. $\sem{B}{}$ is in the reflective subcategory induced by
$T_\ell$. The elimination rule then reduces to Moggi's
interpretation, and we can straightforwardly adapt the soundness
proof, as well as the canonicity proof for $\textsf{Bool}$.

It is now easy to see that the levelled structure on
$\mathbf{CSet}_\mathcal{L}$ from \S\ref{sec:coh2} is a model of
the DCC, with \[
  T_\ell \myeq \blacklozenge_{\downarrow \ell}
\] where $\downarrow \ell \myeq \setcomp{ \ell' \in \mathcal{L} }{
\ell' \sqsubseteq \ell}$ is the \emph{principal lower set} of
$\ell$. It is not hard to show that if $B$ is protected at $\ell$
(as a type) then $\sem{B}{}$ is protected at $\downarrow \ell$ (as
a classified set). As shown in
\S\ref{sec:(co)refl}--\ref{sec:strong}, $\blacklozenge_\pi$ is a
strong, strictly idempotent monad, so that $\sem{B}{} =
\blacklozenge_{\downarrow\ell} \sem{B}{}$.

There is no explicit noninterference theorem stated for DCC itself
in \cite{Abadi1999}. Instead, there are six translations from
various calculi into DCC, which are used to prove noninterference
for each of these `source' calculi. The technique is always the
same: first, show that the translation $(-)^\dagger$ is
\emph{adequate}, in the sense that $E$ has a canonical form in the
source calculus if and only if the semantics of the translation
satisfies $\sem{E^\dagger}{} \neq \bot$. Then, some argument
similar to our Prop. \ref{prop:cocoint2} is used to show constancy
of the `DCC-induced' semantics $\sem{E^\dagger}{}$ of a term $E :
\textsf{Bool}$ with a single free variable of an appropriately
`secure' type. Then $\sem{(E[M/x])^\dagger}{} =
\sem{(E[M'/x])^\dagger}{}$ for all $M, M'$, and a single use of
adequacy suffices to complete the argument. 

We will attempt to capture the essence common to these proofs by
the following proposition.

\begin{prop}
  \label{prop:mmnoninterference}
  If $\pi - \pi' \neq \emptyset$, $A$ is non-empty, and $B$ is
  visible at $\pi - \pi'$, then all morphisms \[
    \blacklozenge_\pi A \rightarrow \blacklozenge_{\pi'} B
  \] are constant functions.
\end{prop}
\begin{proof}
  If $B$ is visible at $\pi - \pi'$ then $B = \Box_{\pi  - \pi'}
  B$, so \[
      \blacklozenge_{\pi'} B
    =
      \blacklozenge_{\pi'} \Box_{\pi - \pi'} B
    =
      \Box_{(\pi-\pi')-\pi'} \blacklozenge_{\pi'} B
    =
      \Box_{\pi-\pi'} \blacklozenge_{\pi'} B
    =
      \Delta_{\pi-\pi'} U_{\pi-\pi'} \blacklozenge_{\pi'} B
  \] by Prop. \ref{prop:switch}(9), set theory, and the definition of
  $\Box_\pi$. But also \[
      \blacklozenge_\pi A
    =
      \blacklozenge_{\pi-\pi'} \blacklozenge_{\pi \cap \pi'} A
  \] by Prop. \ref{prop:switch}(4) and set theory. As $\pi - \pi'
  \neq \emptyset$ and $A$ is non-empty, Prop. \ref{prop:cocoint2}
  applies.
\end{proof}

\noindent We also have the following analogue of Lemma
\ref{lem:adeq}.

\begin{lem}[DCC Adequacy]
  \label{lem:dccadeq}
  Suppose we have a sound categorical interpretation for the DCC,
  as described before. Suppose that, as per the definition in
  \citep{Moggi1991}, each monad $T_\ell$ satisfies the \emph{mono
  requirement}, i.e. each component $A \rightarrow T_\ell A$ of
  the unit is mono. Let $G$ be a ground type that satisfies
  canonicity, with an injective interpretation, so that \[
    \sem{ \vdash \textsf{c}_i : G}{}
      =
    \sem{ \vdash \textsf{c}_j : G}{}
      \quad\Longrightarrow\quad
    \textsf{c}_i \equiv \textsf{c}_j
  \] Then this interpretation is \emph{adequate for $T_\ell G$},
  in the sense that \[
    \sem{\ \vdash M : T_\ell G}{}
    = \sem{\ \vdash [\textsf{c}_i]_\ell : T_\ell G}{}
      \quad\Longrightarrow\quad
    \vdash M = [\textsf{c}_i]_\ell : T_\ell G
  \]
\end{lem}
\begin{proof}
  The canonical forms at $T_\ell G$ are exactly $
  [\textsf{c}_i]_\ell$. Let $M$ normalise to
  $[\textsf{c}_j]_\ell$.  Then \[
    \eta_G \circ \sem{\textsf{c}_j}{}
      =
    \sem{[\textsf{c}_j]_\ell}{}
      =
    \sem{M}{} 
      =
    \sem{[\textsf{c}_i]_\ell}{}
      =
    \eta_G \circ \sem{\textsf{c}_i}{}
  \] As $\eta_G$ is mono, $\sem{\textsf{c}_j}{} =
  \sem{\textsf{c}_i}{}$, and hence $\textsf{c}_j \equiv
  \textsf{c}_i$.
\end{proof}

\noindent We can now interpret DCC+booleans into
$\textsf{CSet}_\mathcal{L}$, with $\sem{\textsf{Bool}}{} \myeq
\Updelta\mathbb{B} \cong \mathbf{1} + \mathbf{1}$. This
interpretation satisfies all the requirements of Lemma
\ref{lem:dccadeq}---as every component $A \rightarrow
\blacklozenge_\pi A$ is a mono---so it is adequate.

\begin{thm}[Noninterference for DCC]
  Let $A$ be a non-empty type, and let $x : T_\ell A \vdash M :
  T_{\ell'} \textsf{Bool}$ with $\ell \not\sqsubseteq \ell'$.
  Then, for any $\vdash E, E' : T_\ell$, we have \[
    \vdash M[E/x] = M[E'/x] : T_{\ell'} \textsf{Bool}
  \]
\end{thm}
\begin{proof}
  We have that \[
    \sem{x : T_\ell A \vdash M : T_{\ell'} \textsf{Bool}}{}
      : \blacklozenge_{\downarrow\ell} \sem{A}{}
          \rightarrow \blacklozenge_{\downarrow\ell'}
            \Updelta\mathbb{B}
  \] But $\ell \not\sqsubseteq \ell'$ if and only if
  $\downarrow\ell \not\subseteq \downarrow\ell'$, so
  $\downarrow\ell\ - \downarrow\ell' \neq \emptyset$.  As
  $\Updelta\mathbb{B}$ is visible everywhere, it follows by
  Proposition \ref{prop:mmnoninterference} that $\sem{M}{} :
  \blacklozenge_{\downarrow \ell} \sem{A}{} \rightarrow
  \blacklozenge_{\downarrow \ell'} \Updelta\mathbb{B}$ is a
  constant function, so \[
      \sem{M[E/x]}{}
    =
      \sem{M}{} \circ \sem{E}{}
    =
      \sem{M}{} \circ \sem{E'}{}
    =
      \sem{M[E'/x]}{}
  \] for any $\vdash E, E' : T_{\ell} A$. By adequacy, it follows
  that $\vdash M[E/x] = M[E'/x] : T_{\ell'} \textsf{Bool}$.
\end{proof}

\subsection{The sealing calculus}

The \emph{sealing calculus} was introduced by \citet{Shikuma2008}.
Its history is complicated: it is a simplification of a calculus
introduced by \citet{Tse2004} as a refinement of DCC. The authors
originally hoped to prove noninterference for DCC not through
denotational methods---as in \cite{Abadi1999}---but by translating
it to System F and using parametricity. However, there was a
technical issue in their work. \citet{Shikuma2008} carried out a
similar programme by translating their sealing calculus to simple
types, thus proving noninterference through parametricity for
simple types. Subsequently, \citet{Bowman2015} carried out the
original programme to completion, by translating DCC itself to
System $\text{F}_\omega$.

The sealing calculus is also based on a partial order
$(\mathcal{L}, \sqsubseteq)$ of security levels. It augments the
context of the simply typed $\lambda$-calculus with a finite set
$\pi$ of \emph{observer levels}. Typing judgements are of the form
\[
  \ctxt{\Gamma}{\pi} \vdash M : A
\] The idea is that an observer can only read data the
classification of which is below their observer status. We write
$\ell \sqsubseteq \pi$ to mean that $\ell$ is below some level in
$\pi$.

The introduction rule specifies that a term obtained using
observer access $\ell$ can be \emph{sealed}, thus becoming a term
that can be handled \emph{possibly without} (but not necessarily
without) access $\ell$: \[
  \begin{prooftree}
      \ctxt{\Gamma}{\pi \cup \{\ell\}} \vdash M : A
    \justifies
      \ctxt{\Gamma}{\pi} \vdash [M]_\ell : [A]_\ell
  \end{prooftree}
\] $[A]_\ell$ is a \emph{type sealed at $\ell$}. Conversely, if
the observer level dominates $\ell$, terms can be \emph{unsealed}:
\[
  \begin{prooftree}
      \ctxt{\Gamma}{\pi} \vdash M : [A]_\ell
        \quad
      \ell \sqsubseteq \pi
    \justifies
      \ctxt{\Gamma}{\pi} \vdash M^\ell : A
  \end{prooftree}
\] and, naturally, $\left([M]_\ell\right)^\ell = M$. The rest of
the system is just that of simple types.

Our levelled cohesion can be used quite directly to provide
semantics for the sealing calculus. To do so, notice that the only
rules that interact with the observer context $\pi$ are the
modal/sealing rules: if we forget those for a moment, everything
else is simply-typed $\lambda$-calculus. Thus, we can interpret
the sealing-free/`constant $\pi$' part of the calculus in any
cartesian closed category. We will choose to do so in the
\emph{co-Kleisli category} $\text{CoKl}(\Box_{\downarrow\pi})$ of
$\Box_{\downarrow\pi}$, which has the same objects as
$\mathbf{CSet}_\mathcal{L}$, but whose morphisms $A \rightarrow B$
are the morphisms $\Box_{\downarrow\pi} A \rightarrow B$ of
$\mathbf{CSet}_\mathcal{L}$. This is standard, see e.g. \cite[Ex.
11, \S 10.6]{Awodey2010}. Also standard is the following theorem,
which is considered `folk' by \citet{Brookes1992}, and mentioned
in passing by \citet{Uustalu2008}:
\begin{thm}
  If $Q : \mathcal{C} \longrightarrow \mathcal{C}$ is a
  product-preserving comonad on a cartesian closed category
  $\mathcal{C}$, then its co-Kleisli category $\text{CoKl}(Q)$
  is also cartesian closed.
\end{thm} Thus, a sequent $\ctxt{x_1 : A_1, \dots, x_n : A_n
}{\pi} \vdash M : A$ is interpreted as an arrow \[
  \Box_{\downarrow\pi}(\sem{A_1}{} 
    \times \dots \times 
  \sem{A_n}{})
    \rightarrow 
  \sem{A}{}
\] The idea is that the observer context $\pi$ declassifies
everything at levels below some level in $\pi$. But recall that $\Box_\pi$
is product-preserving, and that in the particular example of
$\mathbf{CSet}_\mathcal{L}$ it is \emph{strictly} product
preserving. Thus, morphisms of the above type are of the form \[
  \Box_{\downarrow\pi}\sem{A_1}{} 
    \times \dots \times 
  \Box_{\downarrow\pi}\sem{A_n}{}
    \rightarrow 
  \sem{A}{}
\] Sealing is uniformly interpreted on types by the redaction
functors: \[
  \sem{[A]_\ell}{} \myeq \blacklozenge_{\downarrow \ell}\
  \sem{A}{}
\] On terms, the sealing and unsealing rules will be interpreted
by using the adjunction $\Box_{\downarrow\pi} \dashv
\blacklozenge_{\downarrow\pi}$ to move between the different
co-Kleisli categories of the comonads $\Box_{\downarrow \pi}$. For
simplicity we explain the case $\ell \not\in \pi$ which discards
$\ell$ from the observer levels. Recall that, by Prop.
\ref{prop:switch}, we have \[
  \Box_{\downarrow(\pi \cup \{\ell\})}
  = \Box_{\downarrow\pi \cup \downarrow\ell}
  = \Box_{\downarrow\pi}\ \Box_{\downarrow\ell}
  = \Box_{\downarrow\ell}\ \Box_{\downarrow\pi}
\] So the interpretation of a term $\ctxt{\Gamma}{\pi \cup
\{\ell\}} \vdash M : A$, which is a morphism $\sem{\Gamma}{}
\rightarrow \sem{A}{}$ in the co-Kleisli category of
$\Box_{\downarrow\pi \cup \{\ell\}}$, is, by the above, really a
morphism of type \[
  \Box_{\downarrow\ell}\left(
    \Box_{\downarrow\pi}\ \sem{A_1}{}
      \times \dots \times
    \Box_{\downarrow\pi}\ \sem{A_n}{}
  \right) \rightarrow \sem{A}{}
\] in $\mathbf{CSet}_\mathcal{L}$. Moving across the adjunction
yields a morphism \[
    \Box_{\downarrow\pi}\ \sem{A_1}{}
      \times \dots \times
    \Box_{\downarrow\pi}\ \sem{A_n}{}
  \rightarrow \blacklozenge_{\downarrow\ell}\sem{A}{}
\] which is now a morphism in the co-Kleisli category of
$\Box_{\pi}$, and we take that to be the interpretation of
$\ctxt{\Gamma}{\pi} \vdash M : [A]_\ell$. Unsealing is obtained
by moving in the opposite direction: doing so, we obtain a
morphism \[
  \Box_{\downarrow\ell}\ \Box_{\downarrow\pi} \sem{A_1}{}
    \times \dots \times
  \Box_{\downarrow\ell}\ \Box_{\downarrow\pi} \sem{A_n}{}
    \rightarrow
  \sem{A}{}
\] But $\ell \sqsubseteq \pi$ implies $\downarrow \pi =\
\downarrow \pi\ \cup \downarrow\ell$, and so
$\Box_{\downarrow\ell}\ \Box_{\downarrow\pi}$ is just
$\Box_{\downarrow\pi}$. Again, it would suffice for these
equalities of modalities to be mere natural isomorphisms.
Soundness of the equations of the sealing calculus then follow
from the fact the adjunction induces a natural isomorphism between
the appropriate hom-sets.

\citet{Shikuma2008} prove the following noninterference theorem.
Say that terms $\ctxt{\cdot}{\pi} \vdash M_1, M_2 : A$ at observer
level $\pi$ are \emph{contextually equivalent} just if there is no
term of ground type at the same observer level that can
distinguish them: that is, if $\ctxt{x : A}{\pi} \vdash N :
\textsf{Bool}$, then $N[M_1/x] = N[M_2/x]$. This defines an
equivalence relation $\approx_\pi$, which extends to substitutions
$\ctxt{\cdot}{\pi} \vdash \sigma : \Gamma$. The main theorem
states that if $\sigma \approx_\pi \sigma' : \Gamma $ and
$\ctxt{\Gamma}{\pi} \vdash E : A$, then $E[\sigma] \approx_\pi
E[\sigma']$. That is: substituting terms indistinguishable at
$\pi$ yields results indistinguishable at $\pi$.

This is a particularly strong noninterference theorem, which takes
many pages of very beautiful---but also painfully
elaborate!---work to show. We will content ourselves with using
the model to provide the following direct corollary: \footnote{It
is easy to show that any $M, N : [A]_\ell$ are contextually
equivalent at $\pi \not\sqsupseteq \ell$ by a logical relations
argument: see \cite[Theorem 2.18]{Shikuma2008}.  Then this result
is the special case of the noninterference theorem, once we
observe that equality and observational equivalence coincide at
the ground type $\textsf{Bool}$.}

\begin{thm}[Noninterference for Sealing Calculus]
  Let $A$ be a non-empty type.  If $\ell \not\sqsubseteq \pi$,
  then for any $\ctxt{\cdot}{\pi} \vdash M, N : [A]_\ell$ and
  $\ctxt{x : [A]_\ell}{\pi} \vdash E : \textsf{Bool}$, we have
  $\ctxt{\cdot}{\pi} \vdash E[M/x] = E[N/x] : \textsf{Bool}$.
\end{thm}

\noindent First, we notice that \citet{Shikuma2008} show
confluence and strong normalisation for the sealing calculus
(including unit and coproducts, which subsume $\textsf{Bool}$).
Thus canonicity holds, and Lemma \ref{lem:adeq} applies to yield
adequacy at $\textsf{Bool}$. The domain of $\sem{\ctxt{x :
[A]_\ell}{\pi} \vdash E : \textsf{Bool}}{}$ is
$\blacklozenge_{\downarrow\ell} \sem{A}{}$ as an object of
$\text{CoKl}(\Box_{\downarrow\pi})$, and hence as an object of
$\textbf{CSet}_\mathcal{L}$ it is \[
  \Box_{\downarrow\pi}\ \blacklozenge_{\downarrow\ell} \sem{A}{}
  = \blacklozenge_{\downarrow\ell - \downarrow\pi}\
  \Box_{\downarrow\pi} \sem{A}{}
\] by Prop. \ref{prop:switch}(8). Its codomain is
$\Delta\mathbb{B}$, which is visible everywhere. But, as $\ell
\not\sqsubseteq \pi$, we have that $\downarrow\ell\ -
\downarrow\pi \neq \emptyset$, so Prop.
\ref{prop:mmnoninterference} applies (with $\pi' = \emptyset$) to
show that it is a constant function. Then, \[
  \sem{E[M/x]}{} 
    = \sem{E}{} \circ_\text{CoKl} \sem{M}{}
    = \sem{E}{} \circ_\text{CoKl} \sem{N}{} 
    = \sem{E[N/x]}{}
\] where $\circ_\text{CoKl}$ is composition
$\text{CoKl}(\Box_{\downarrow\pi})$, and by using adequacy the
proof is complete.

\section{Conclusion}
  \label{sec:conc}

To recapitulate: we have defined the model of classified sets, and
shown that it forms a pre-cohesion. This led us to the generation
of modalities $\int \dashv \Box \dashv \blacklozenge$, and the
proof of noninterference properties for Moggi's monadic
metalanguage and the Davies-Pfenning comonadic calculus. Next, we
took a levelled view of cohesion, and showed that this generates a
multi-modal framework $\int_\pi \dashv \Box_\pi \dashv
\blacklozenge_\pi$. These modalities satisfy many algebraic laws,
which we then used to prove noninterference for two multi-modal
information flow calculi, the dependency core calculus, and the
sealing calculus. 

We discuss two aspects of our work that we believe may lead to
interesting future developments. 

\paragraph{Cohesion as a theory of information flow}

Our results demonstrate that the very general and abstract
framework of pre-cohesion in fact has very concrete applications
in analysing information flow. Our noninterference proofs rely on
very general lemmas about pre-cohesion---in the case of
\S\ref{sec:noninterference2}, with some additional, slightly
mysterious equations---which are then applied to each calculus by
using some usually very simple form of adequacy. We believe
these to be a simplification compared to previous work on
noninterference, which required quite a bit of hard work in terms
of fully abstract/fully complete translations, e.g.
\cite{Shikuma2008, Bowman2015}.

However, we believe this to be the tip of the iceberg: cohesion
can tell us much more about information flow. However, this cannot
happen unless we replace adequacy, which is a very weak form of
completeness, with something stronger. It would be very
interesting to see whether there are information flow calculi for
which classified sets are \emph{fully abstract}, in the sense that
$M$ and $N$ are contextually equivalent if and only if $\sem{M}{}
= \sem{N}{}$ in $\textbf{CSet}_\mathcal{L}$. That would grant us
far more power to use ideas from and properties of cohesion to
prove theorems about information flow.

Conversely, we can seek information flow calculi with likeness to
cohesion-based models: we can look at what structure is available
in the multi-modal setting of classified sets, and then try to
formulate a calculus from that. This would most likely lead to a
multi-modal version of the \emph{spatial type theory} of
\citet{Shulman2018}. Such a calculus would also have much to offer
in terms of resolving the debate between, for example,
coarse-grained and fine-grained formulations, as discussed by
\citet{Rajani2018}. In a sense, its formulation would provide a
mathematical justification for \emph{canonical}, type-theoretic
choices of modalities and constructs. (That is, if we of course
accept classified sets as a canonical model of information flow.)
Additionally, the calculus would include the \emph{shape at $\pi$}
modality ($\int_\pi$), which has never appeared before in papers
on information flow type systems, and which might have interesting
applications as a `security quotient' or `secure view' type
constructor. Finally, there seem to be close connections between
this work and \emph{graded monads and comonads}, which also have
applications in information flow: see \cite{Gaboardi2016}.

\paragraph{Cohesion as a basis for multi-modal type theories}

The other side of the coin in the present development is that
information flow---which is a garden-variety application for
multi-modal types---can be quite eloquently spoken about in this
language of cohesion. In particular, the formulation of a functor
$\mathcal{P}(\mathcal{L})^\text{op} \longrightarrow
\textbf{Precoh}$ enabled very short and conceptual proofs of
results that would ordinarily require a lot of uninteresting
relational reasoning. It is thus natural to ask whether there
might be other functors $\mathcal{C}^\text{op} \longrightarrow
\textbf{Precoh}$ of \emph{$\textbf{Precoh}$-valued presheaves}
that yield interesting analyses of multi-modal logical systems. In
fact, some of the notation we developed in \S\ref{sec:coh2} bears
a striking similarity to the adjoint logic of \citet{Licata2016},
and the framework of \citet{Licata2017}. Could there be a closer
connection between these developments?

The combination of the parametricity-style reasoning and cohesion
has also recently appeared in the multi-modal type theories of
\citet{Nuyts2017} and \citet{Nuyts2018}. There is certainly
potential for a very interesting connection to be made there, e.g.
by devising a dependent type theory for information flow, for
which the noninterference theorems can be proven
\emph{internally}.

\paragraph{Related work}

As mentioned in the introduction, the problem of information flow
is almost as old as Computer Science itself, dating at least as
far back as Bell and LaPadula's report of 1973 \cite{Bell1996}.
The notion of \emph{noninterference} itself was introduced by
\citet{Goguen1982}. The use of types to guarantee information flow
control, and hence some form of noninterference, appears to have
begun in the 1990s, with the first work on higher-order functional
programming being that of \citet{Heintze1998} on the SLam
calculus; see \emph{op. cit.} for a useful list of references to
approaches that preceded it. \citet{Rajani2018} provide a good
overview and a largely complete list of references to the
literature thereafter.

Even the first higher-order noninterference results, such as those
for SLam \cite{Heintze1998}, use some form of logical relations,
who directly state that they are ``borrowing ideas from
Reynolds.'' In particular, \citet{Heintze1998} use logical
relations on top of a denotational model, which then led to the
\emph{dependency category} of \citet{Abadi1999}, which we have
refined into classified sets. Other than that, there appears to be
very little other work on denotational models of information flow;
even if situated outside the higher-order functional setting, we
ought to mention the work of \citet{Sabelfeld2001}.

Furthermore, we should note that this paper is the first
\emph{categorical} approach to information flow. Even though we
mostly use the motivating example of classified sets, all our
theorems are rather general, and apply to all pre-cohesive
settings indexed over subsets of $\mathcal{L}$ for which the
equations \ref{eqn:1}, \ref{eqn:2} and \ref{eqn:3} apply. It is
also worth noting that we did not use any particular structure on
$\mathcal{L}$, even though most of the related work cited above
asks for some kind of lattice---even if they do not use it either:
we merely made use of indexing over $\mathcal{P}(\mathcal{L})$.

%% Acknowledgments
\begin{acks}
  I would like to thank Dan Licata for numerous observations that
  led to the material in this paper. The pasting diagrams in
  \S\ref{sec:stack} are due to Amar Hadzihasanovic. Thanks are due
  to Mario Alvarez-Picallo, Mike Shulman, and the anonymous
  reviewers for their many useful suggestions, corrections, and
  careful reading. Finally, I would also like to thank Dan
  Licata's cat, Otto, for keeping me company during the writing of
  this paper.

  This material is based upon work supported by the
  \grantsponsor{}{Air Force Office of Scientific Research}{} under
  award number \grantnum{}{FA9550-16-1-0292}. Any opinions,
  finding, and conclusions or recommendations expressed in this
  material are those of the author(s) and do not necessarily
  reflect the views of the United States Air Force.
\end{acks}

%% Bibliography
\input{classifieds.bbl}

%\bibliography{/Users/alex/cs/lib/library.bib}

\end{document}

%% file: classifieds.bbl
%%% -*-BibTeX-*-
%%% Do NOT edit. File created by BibTeX with style
%%% ACM-Reference-Format-Journals [18-Jan-2012].